\documentclass[a4paper,UKenglish,cleveref, autoref]{lipics-v2019}

\title{On the Parameterized Complexity of Grid Contraction} %TODO Please add

\titlerunning{On the Parameterized Complexity Of Grid Contraction}%optional, please use if title is longer than one line

\author{Saket Saurabh}{The Institute Of Mathematical Sciences, HBNI, Chennai, India \\ University of Bergen, Bergen, Norway}{saket@imsc.res.in}{}{
This project has received funding from the European Research Council
\begin{minipage}{0.6\textwidth}
(ERC) under the European Union's Horizon $2020$ research and innovation programme (grant agreement No $819416$), and Swarnajayanti Fellowship (No DST/SJF/MSA01/2017-18).
\end{minipage}
\begin{minipage}{0.3\textwidth}
        \includegraphics[scale=1]{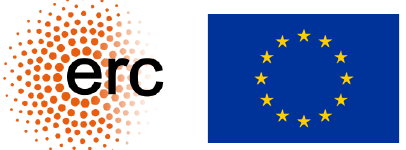}
\end{minipage}
}
\author{Uéverton dos Santos Souza}{Fluminense Federal Universidade, Niterói, Brazil}{ueverton@ic.uff.br}{}{}
\author{Prafullkumar Tale}{Max Planck Institute for Informatics, Saarland Informatics Campus,
Saarbr\"ucken, Germany.}{prafullkumar.tale@mpi-inf.mpg.de}{}{This research is a part of a project that has received funding from the European Research Council (ERC) under the European Union's Horizon $2020$ research and innovation programme under grant agreement SYSTEMATICGRAPH (No. $725978$).}

\authorrunning{Saurabh, Souza, and Tale}
\Copyright{Saket Saurabh, Uéverton dos Santos Souza, and Prafullkumar Tale}%TODO mandatory, please use full first names. LIPIcs license is "CC-BY";  http://creativecommons.org/licenses/by/3.0/

\ccsdesc[500]{Theory of computation~Fixed parameter tractability}

\keywords{Grid Contraction, FPT, Kernelization, Lower Bound}%TODO mandatory; please add comma-separated list of keywords

\category{}%optional, e.g. invited paper

\relatedversion{}%optional, e.g. full version hosted on arXiv, HAL, or other respository/website
\supplement{}%optional, e.g. related research data, source code, ... hosted on a repository like zenodo, figshare, GitHub, ...

\nolinenumbers %uncomment to disable line numbering

\EventEditors{John Q. Open and Joan R. Access}
\EventNoEds{2}
\EventLongTitle{42nd Conference on Very Important Topics (CVIT 2016)}
\EventShortTitle{CVIT 2016}
\EventAcronym{CVIT}
\EventYear{2016}
\EventDate{December 24--27, 2016}
\EventLocation{Little Whinging, United Kingdom}
\EventLogo{}
\SeriesVolume{42}
\ArticleNo{23}
\usepackage[linesnumbered,algoruled,boxed,lined]{algorithm2e}
\usepackage{complexity} % For symbols in complexity class.
\usepackage{todonotes} % For todo notes
\presetkeys{todonotes}{inline}{}

\newcommand{\grid}{\boxplus}

\newcommand{\al}{\alpha}
\newcommand{\bt}{\beta}

\newcommand{\calA}{\mathcal{A}}

\newcommand{\calG}{\mathcal{G}}
\newcommand{\calH}{{\mathcal H}}

\newcommand{\calO}{\ensuremath{{\mathcal O}}}
\newcommand{\OO}{\mathcal{O}}

\newcommand{\calP}{\mathcal{P}}

\newcommand{\calV}{\mathcal{V}}
\newcommand{\calW}{\mathcal{W}}
\newcommand{\WC}{W^c}

\newcommand{\true}{\texttt{True}}
\newcommand{\false}{\texttt{False}}

\newcommand{\ETH}{\textsf{ETH}}

\newcommand{\yes}{\textsc{Yes}}
\newcommand{\no}{\textsc{No}}

\newtheorem{observation}{Observation}[section]
\newtheorem{reduction rule}{Reduction Rule}[section]
\newtheorem{marking-scheme}{Marking Scheme}[section]
\newcommand{\defparproblem}[4]{
  \vspace{1mm}
\noindent\fbox{
  \begin{minipage}{0.96\textwidth}
  \begin{tabular*}{\textwidth}{@{\extracolsep{\fill}}lr} #1  & {\bf{Parameter:}} #3
\\ \end{tabular*}
  {\bf{Input:}} #2  \\
  {\bf{Question:}} #4
  \end{minipage}
  }
  \vspace{1mm}
}

\begin{document}

\maketitle

\begin{abstract}
%Add abstract
For a family of graphs $\mathcal{G}$, the $\mathcal{G}$-\textsc{Contraction} problem takes as an input a graph $G$ and an integer $k$, and the goal is to decide if there exists $F \subseteq E(G)$ of size at most $k$ such that $G/F$ belongs to $\mathcal{G}$.
Here, $G/F$ is the graph obtained from $G$ by contracting all the edges in $F$.
In this article, we initiate the study of \textsc{Grid Contraction} from the parameterized complexity point of view.
We present a fixed parameter tractable algorithm,  running in time $c^k \cdot |V(G)|^{\calO(1)}$, for this problem.
We complement this result by proving that unless \ETH\ fails, there is no algorithm for \textsc{Grid Contraction} with running time $c^{o(k)} \cdot |V(G)|^{\calO(1)}$.
We also present a polynomial kernel for this problem.

\end{abstract}

%\newpage
%\setcounter{page}{1}
\section{Introduction}
Graph modification problems are one of the central problems in graph
theory that have received a lot of attention in theoretical computer
science. Some of the important graph modification operations are
vertex deletion, edge deletion, and edge contraction.
For graph $G$, any graph that can be obtained from $G$ by using these
three types of modifications is called a \emph{minor} of $G$.
If only the first two types of modification operations are allowed then resulting
graph is said to a \emph{subgraph} of $G$.
If the only third type of
modification is allowed then the resulting graph is called a
\emph{contraction} of $G$.

For two positive integer $r, q$, the $(r \times q)$-grid is a graph in
which every vertex is assigned a unique pair of the form $(i, j)$ for $1 \le i \le r$ and $1 \le j \le l$.
A pair of vertices $(i_1, j_1)$ and $(i_2, j_2)$ are adjacent with
each other if and only if $|i_1 - i_2| + |j_1 - j_2| = 1$.
There has been considerable attention to the problem of obtaining a
grid as a minor of the given graph.
We find it surprising that the very closely related question of obtaining a grid as a contraction did not receive any attention.
In this article, we initiate a study of this problem from the parameterized complexity point of view.

The {\em contraction} of edge $uv$ in simple graph $G$ deletes vertices $u$ and $v$ from $G$, and replaces them by a new vertex, which is made adjacent to vertices that were adjacent to either $u$ or $v$.
Note that the resulting graph does not contain self-loops and multiple edges.
A graph $G$ is said to be \emph{contractible} to graph $H$ if $H$ can be obtained from $G$ by edge contractions.
Equivalently, $G$ is contractible to $H$ if $V(G)$ can be partitioned into $|V(H)|$ many connected sets, called witness sets, and these sets can be mapped to vertices in $H$ such that adjacency between witness sets is consistent with their mapped vertices in $H$.
If such a partition of $V(G)$ exists then we call it $H$-witness structure of $G$.
A graph $G$ is said to be \emph{$k$-contractible} to $H$ if $H$ can be obtained from $G$
by $k$ edge contractions.
For a family of graphs $\calG$, the $\calG$-\textsc{Contraction} problem
takes as an input a graph $G$ and an integer $k$, and the objective is
to decide if $G$ is $k$-contractible to a graph $H$ in $\calG$.

\vspace{0.1cm}
\noindent \textbf{Related Work :}
Early papers of Watanabe et al. \cite{contractEarly2, contractEarly3},
Asano and Hirata \cite{asano1983edge} showed
$\calG$-\textsc{Contraction} is \NP-Complete for various class of
graphs like planar graphs, outer-planar graphs, series-parallel graphs, forests, chordal graphs.
Brouwer and Veldman proved that it is \NP-Complete even to determine
whether a given graph can be contracted to a path of length four or
not \cite{brouwer1987contractibility}.
In the realm of parameterized complexity, $\calG$-\textsc{Contraction}
has been studied with the parameter being the number of edges allowed to
be contracted.
It is known that $\calG$-\textsc{Contraction} admits an \FPT\
algorithm when $\calG$ is set of paths \cite{tree-contraction}, trees \cite{tree-contraction}, cactus
\cite{krithika2018fpt}, cliques \cite{cai2013contracting}, planar
graphs \cite{planarContract} and bipartite graphs
\cite{bipartiteContract, bipartiteContract2}.
For a fixed integer $d$, let $\calH_{\ge d}, \calH_{\le d}$ and $\calH_{=d}$ denote the set of graphs with minimum degree at least $d$, maximum degree at most $d$, and $d$-regular graphs, respectively.
Golovach et al.~\cite{golovach2013increasing} and Belmonte et al.~\cite{Belmonte:2014} proved that $\calG$-\textsc{Contraction} admits an \FPT\ algorithm when $\calG \in \{\calH_{\ge d}, \calH_{\le d}, \calH_{= d}\}$.
When $\calG$ is split graphs or chordal graphs, the $\calG$-\textsc{Contraction} is known to be $\W[1]$-hard~\cite{agrawal2019split} and
$\W[2]$-hard \cite{elimNew, cai2013contracting}, respectively.
To the best of our knowledge, it is known that $\calG$-\textsc{Contraction} admits a polynomial kernel only when $\calG$ is a set of paths \cite{tree-contraction} or set of paths or cycle i.e. $\calH_{\le 2}$ \cite{Belmonte:2014}.  
It is known that $\calG$ does not admit a polynomial kernel, under standard complexity assumptions, when $\calG$ is set of trees~\cite{tree-contraction}, cactus~\cite{lossy-fst}, or cliques~\cite{cai2013contracting}.

\vspace{0.1cm}
\noindent \textbf{Our Contribution :}
In this article we study parameterized complexity of \textsc{Grid
  Contraction} problem. We define the problem as follows.
  
\defparproblem{\textsc{Grid Contraction}}{Graph $G$ and integer $k$}{$k$}{Is $G$ $k$-contractible to a grid?}

To the best of our knowledge, the computation complexity of the problem is
not known nor it is implied by the existing results regarding edge
contraction problems.
We prove that the problem is indeed \NP-Complete (Theorem~\ref{thm:lower-bound}). 
We prove that there exists an \FPT\ algorithm which given an instance
$(G, k)$ of \textsc{Grid Contraction} runs in time $4^{6k} \cdot
|V(G)|^{\calO(1)}$ and correctly concludes whether it is a \yes\
instance or not (Theorem~\ref{theorem:grid-contraction}).
We complement this result by proving that unless \ETH\ fails there is no algorithm for \textsc{Grid Contraction} with running time $2^{o(k)}\cdot |V(G)|^{\calO(1)}$ (Theorem~\ref{thm:lower-bound}).
We present a polynomial kernel with $\calO(k^4)$ vertices and edges for \textsc{Grid Contraction} (Theorem~\ref{thm:kernel-grid}).

\vspace{0.1cm}
\noindent \textbf{Our Methods :} Our \FPT\ algorithm for \textsc{Grid
  Contraction} is divided into two phases. In the first phase, we introduce a
restricted version of \textsc{Grid Contraction} problem called \textsc{Bounded Grid Contraction}.
In this problem, along with a graph $G$ and an integer $k$, an input
consists of an additional integer $r$.
The objective is to determine whether graph $G$ can be
$k$-contracted to a grid with $r$ rows.
We present an \FPT\ algorithm parameterized by $(k + r)$ for this problem.
This algorithm is inspired by the exact exponential algorithm for
\textsc{Path Contraction} in \cite{path-contraction}.
It is easy to see that an instance $(G, k)$ is a \yes\ instance of
\textsc{Grid Contraction} if and only if $(G, k, r)$ is a \yes\
instance of \textsc{Bounded Grid Contraction} for some $r$ in $\{1, 2,
\dots, |V(G)|\}$.
In the second phase, given an instance $(G, k)$ of
\textsc{Grid Contraction} we  produce
polynomially many instances of \textsc{Bounded Grid Contraction} such
that -- $(a)$ the input instance is a \yes\ instance if and only if at least one of the produced instances is
a \yes\ instance and
$(b)$ for any produced instance, say $(G', k', r)$, we have $k' = k$ and $r \in \{1, 2,
\dots, 2k + 5\}$.
We prove that all these instances can be produced in time polynomial
in the size of the input.
An \FPT\ algorithm for \textsc{Grid
Contraction} is a direct consequence of these two results.
We use techniques presented in the second phase to obtain a polynomial kernel
for \textsc{Grid Contraction}.

We present a brief overview of the \FPT\ algorithm for \textsc{Bounded
  Grid Contraction}.
\emph{Boundary vertices} of a subset $S$ of $V(G)$ are the vertices
in $S$ which are adjacent to at least one vertex in $V(G) \setminus
S$.
A subset $S$ of $V(G)$ is \emph{nice} if both $G[S], G - S$ are
connected, and $G[S]$ can be contracted to a $(r\times q)$-grid with
all boundary vertices in $S$ in an end-column for some integer $q$.
In other words, a subset $S$ of $V(G)$ is nice if it is a union of witness
sets appearing in first few columns in some grid witness structure of
$G$.
See Definition~\ref{def:nice-subset}.
The objective is to keep building a \emph{special partial solution} for some nice subsets. 
In this special partial solution, all boundary vertices of a
particular nice subset are contained in bags appearing in an end-column. 
This partial solution is then extended to the remaining graph.
The central idea is -- \emph{for a nice subset $S$ of graph $G$, if $G[S]$ can be contracted to a grid such that all boundary vertices of $S$ are in an end bag then how one contract $G[S]$ is irrelevant}.
This allows us to store one solution for $G[S]$ and build a dynamic
programming table nice subsets of vertices.
The running time of such an algorithm depends on the following two quantities $(i)$ the number of possible entries in the dynamic programming table, and $(ii)$ time spent at each entry.
We prove that \emph{to bound both these quantities as a function of $k$, it is sufficient to know the size of neighborhood of $S$ and the size of the union of witness sets in an end-column in a grid contraction of $G[S]$ which contains all boundary vertices of $S$}.

In the second phase, we first check whether a given graph $G$ can be
$k$-contracted to a grid with $r$ rows for $r \in \{1, 2, \dots, 2k +
5\}$ using the algorithm mentioned in the previous paragraph.
If for any value of $r$ it returns \yes\ then we can conclude that
$(G, k)$ is a \yes\ instance of \textsc{Grid Contraction}.
Otherwise, we argue that there exists a \emph{special} separator $S$ in $G$ which
induces a $(2 \times q)$ grid for some positive integer $p$.
We prove that it is safe to contract $q$ vertical edges in $G[S]$.
Let $G'$ be the graph obtained from $G$ by contracting these parallel edges.
Formally, we argue that $G$ is $k$-contractible to a $(r' \times
q)$-grid if and only if $G'$ is $k$-contractible to $((r' - 1) \times
q)$-grid.
We keep repeating the process of finding a special separator and
contracting parallel edges in it until one of the following things
happens --
$(a)$ The resultant graph is $k$-contractible to a $(r' \times
q)$-grid for some $r' < 2k + 5$.
$(b)$ The resultant graph does not contain a special separator.
We argue that in Case~$(b)$, it is safe to conclude that $(G, k)$ is a
\no\ instance for \textsc{Grid Contraction}.

\vspace{0.1cm}
\noindent \textbf{Organization of the paper :} We present some
preliminary notations which will be used in rest of the paper in
Section~\ref{sec:prelims}.
We present a crucial combinatorial lemma in
Section~\ref{sec:comb_lemma}.
As mentioned earlier, this algorithm is divided into two phases.
We present the first and the second phase in
Section~\ref{sec:fpt-bounded-grid} and \ref{sec:fpt-grid}, respectively.
Section~\ref{sec:fpt-grid} also contains an \FPT\ algorithm for \textsc{Grid Contraction}.
We prove that the dependency on the parameter in the running time of
this algorithm is optimal, up to a constant factor, unless ETH fails in
Section~\ref{sec:np-complete-lower-bound}.
In Section~\ref{sec:kernel}, we present a polynomial kernel for
\textsc{Grid Contraction} problem.
We conclude the paper with an open question in Section~\ref{sec:conclusion}.
\section{Preliminaries}
\label{sec:prelims}
For a positive integer $k$, $[k]$ denotes the set $\{1,2,\ldots, k\}$.

\subsection{Graph Theory}
In this article, we consider simple graphs with a finite number of vertices.
%We use standard notation from graph theory \cite{diestel-book}. 
For an undirected graph $G$, sets $V(G)$ and $E(G)$ denote its set of vertices and edges respectively. Two vertices $u, v$ in $V(G)$ are said to be \emph{adjacent} if there is an edge $uv$ in $E(G)$. 
The neighborhood of a vertex $v$, denoted by $N_G(v)$, is the set of vertices adjacent to $v$ and its degree $d_G(v)$ is $|N_G(v)|$.
The subscript in the notation for neighborhood and degree is omitted if the graph under consideration is clear.
For a set of edges $F$, set $V(F)$ denotes the collection of endpoints of edges in $F$.
For a subset $S$ of $V(G)$, we denote the graph obtained by deleting $S$ from $G$ by $G - S$ and the subgraph of $G$ induced on the set $S$ by $G[S]$. 
For two subsets $S_1, S_2$ of $V(G)$, we say $S_1, S_2$ are adjacent if there exists an edge with one endpoint in $S_1$ and other in $S_2$. 
For a subset $S$ of $V(G)$, let $\Phi(S)$ denotes set of vertices in $S$ which are adjacent with at least one vertex outside $S$. Formally, $\Phi(S) = \{s \in S |\ N(s) \setminus S \neq \emptyset\}$. These are also called \emph{boundary vertices} of $S$.

A {\em path} $P=(v_1,\ldots,v_l)$ is a sequence of distinct vertices where every consecutive pair of vertices is adjacent.
For two vertices $v_1, v_2$ in $G$, $dist(v_1, v_l)$ denotes the length of a shortest path between these two vertices.
A graph is called {\em connected} if there is a path between every pair of distinct vertices. It is called {\em disconnected} otherwise.
A set $S$ of $V(G)$ is said to be a \emph{connected set} if $G[S]$ is connected.
For two vertices $v_1, v_2$ in $G$, a set $S$ is called ($v_1$-$v_2$)-\emph{separator}, if any $v_1$-$v_2$ paths intersects $S$. 
If a set is a ($v_1$-$v_2$)-separator as well as ($v_3$-$v_4$)-separator then we write it as \{($v_1$-$v_2$), ($v_3$-$v_4$)\}-separator.

For two positive integer $r, q$, the $(r \times q)$-grid is a graph on $r\cdot q$ vertices.
The vertex set of this graph consists of all pairs of the form $(i, j)$ for $1 \le i \le r$ and $1 \le j \le q$.
A pair of vertices $(i_1, j_1)$ and $(i_2, j_2)$ are adjacent with each other if and only if $|i_1 - i_2| + |j_1 - j_2| = 1$.
We say that such a graph is a grid with $r$ rows and $q$ columns.
It is called a $(r \times q)$-grid and is denoted by $\grid_{r \times q}$.
We use $\grid$ to denote a grid with an unspecified number of rows and columns. 
The vertices in grid $\grid$ are denoted by $\grid[i, j]$ or simply by $[i, j]$.
Note that the grid with exactly one row is a path.
To remove some corner cases, we consider grids that have at least two rows and two columns. 
Any grid contains exactly four vertices that have degree two. These vertices are called \emph{corner vertices}. 
Let $t_1 = [1, 1]$, $t_2 = [1, q]$, $t_3 = [r, q]$, and $t_4 = [r, 1]$ be the corner vertices in grid $\grid_{r \times q}$.

\begin{observation}\label{obs:corner-sep-size} If $\hat{S}$ is a connected \{($t_1$-$t_4$), ($t_2$-$t_3$)\}-separator in $\grid_{r \times q}$ then its size is at least $q$. Moreover, if $|\hat{S}| = q$ then it corresponds to a row in $\grid_{r \times q}$. 
\end{observation}
\begin{proof} 
Without loss of generality, we assume that $t_1 \equiv [1, 1], t_2 \equiv [1, q], t_3 \equiv [r, q]$, and $t_4 \equiv [r, 1]$. 
Consider the $(t_1-t_4)$ path which contains vertices in the first column of the grid and the $(t_2-t_3)$ path which contains all vertices in the last column.
Since $\hat{S}$ separates $t_1$ from $t_4$ and $t_2$ from $t_3$, it contains at least two vertices of the form $[i, 1]$ and $[i', q]$ for some $i, i'$ in $\{1, 2, \dots, r\}$. 
Since $\hat{S}$ is connected then it contains a path connecting these two vertices. 
As any path connecting these two vertices contains at least $q - 1 + |i - i'| + 1 \le q$ vertices, the size of $\hat{S}$ is at least $q$. 
If $|\hat{S}| = q$ then $|i - i'| = 0$. 
This implies two endpoints of a row are contained in $\hat{S}$. 
Since $\hat{S}$ is of size $q$, vertices that are present in $\hat{S}$ are from one row.
This proves the second part of the observation.
\end{proof}

\subsection{Graph Contraction}

The {\em contraction} of edge $uv$ in $G$ deletes vertices $u$ and $v$ from $G$, and adds a new vertex, which is made adjacent to vertices that were adjacent to either $u$ or $v$.
Notice that no self-loop or parallel edge is introduced in this process.
The resulting graph is denoted by $G/e$.
For a given graph $G$ and edge $e = uv$, we formally define $G/e$ in the following way: $V(G/e) = (V(G) \cup \{w\}) \backslash \{u, v\}$ and $E(G/e) = \{xy \mid x,y \in V(G) \setminus \{u, v\}, xy \in E(G)\}  \cup \{wx |\ x \in N_G(u) \cup N_G(v)\}$.
Here, $w$ is a new vertex which was not in $V(G)$.
An edge contraction reduces the number of vertices in a graph by exactly one.
Several edges might disappear due to one edge contraction. 
%For a subset of edges $F$ in $G$, graph $G/ F$ denotes the graph obtained from $G$ by repeatedly contracting edge in $F$ until no such edges remain.
For a subset of edges $F$ in $G$, graph $G/ F$ denotes the graph obtained from $G$ by contracting each connected component in the sub-graph $G' = (V(F), F)$ to a vertex.

\begin{definition}[Graph Contraction] \label{def:graph-contractioon} A graph $G$ is said to be \emph{contractible} to graph $H$ if there exists an onto function $\psi: V(G) \rightarrow V(H)$ such that following properties hold.
\begin{itemize}
\item For any vertex $h$ in $V(H)$, graph $G[W(h)]$ is connected and not empty, where set $W(h) := \{v \in V(G) \mid \psi(v)= h\}$.
\item For any two vertices $h, h'$ in $V(H)$, edge $hh'$ is present in $H$ if and only if there exists an edge in $G$ with one endpoint in $W(h)$ and another in $W(h')$.
\end{itemize}
\end{definition}

We say graph $G$ is contractible to $H$ via mapping $\psi$.
For a vertex $h$ in $H$, set $W(h)$ is called a \emph{witness set} associated with/corresponding to $h$. 
We define $H$-\emph{witness structure} of $G$, denoted by $\mathcal{W}$, as a collection of all witness set.
Formally, $\mathcal{W}=\{W(h) \mid h \in V(H)\}$.
A witness structure $\mathcal{W}$ is a partition of vertices in $G$. 
If a \emph{witness set} contains more than one vertex then we call it \emph{big} witness-set, otherwise it is \emph{small/singleton} witness set. 

If graph $G$ has a $H$-witness structure then graph $H$ can be obtained from $G$ by a series of edge contractions.
For a fixed $H$-witness structure, let $F$ be the union of spanning trees of all witness sets.
By convention, the spanning tree of a singleton set is an empty set. To obtain graph $H$ from $G$, it is necessary and sufficient to contract edges in $F$.
We say graph $G$ is \emph{$k$-contractible} to $H$ if cardinality of $F$ is at most $k$.
In other words, $H$ can be obtained from $G$ by at most $k$ edge contractions.
The following observations are immediate consequences of definitions.
\begin{observation}
  \label{obs:witness-structure-property} If graph $G$ is $k$-contractible to graph $H$ via mapping $\psi$ then following statements are true.  
\begin{enumerate}
\item $|V(G)| \leq |V(H)|+ k$. 
%\item For any witness set $W$ in a $H$-witness structure of $G$, cardinality of $W$ is at most $k+1$.
\item Any $H$-witness structure of $G$ has at most $k$ big witness sets.
\item For a fixed $H$-witness structure, the number of vertices in $G$ which are contained in big witness sets is at most $2k$.
\item If $S$ is a $(x_1-x_2)$-separator in $G$ then $\psi(S)$ is a $(\psi(s_1)-\psi(s_2))$-separator in $H$.
\item If $S$ is a separator in $G$ such that there are at least two connected components of $G \setminus S$ which has at least $k + 1$ vertices, then $\psi(S)$ is a separator in $H$.
\end{enumerate}
\end{observation}
\begin{proof} The proof of $(1), (2)$ and $(3)$ follows directly from the definitions.

\vspace{0.2cm}
\noindent $(4)$ Consider any ($x_1-x_2$)-path $P$ in $G$.
Note that $\psi(P)$ corresponds to a $(\psi(x_1)$-$\psi(x_2))$-path in $H$ (with possible repetition of vertices).
Since $S$ is a $(x_1-x_2)$-separator, every $(x_1-x_2)$-path intersects with $S$. This implies that $\psi(P)$ intersects $\psi(S)$. Since $P$ is an arbitrary $(x_1-x_2)$-path in $G$, we can conclude that every $(\psi(x_1) - \psi(x_2))$-path in $H$ intersects $\psi(S)$. Hence $\psi(S)$ is a $(\psi(s_1)-\psi(s_2))$-separator in $H$. 

\vspace{0.2cm}
\noindent $(5)$ Let $C_a$ and $C_b$ be two connected components of $G \setminus S$ which has at least $k + 1$ vertices.
Since, $G$ is $k$-contractible to $H$, there exists a vertex $v_a$ in $C_a$ (similarly, $v_b$ in $C_b$) such that $\psi(v_a) \neq \psi(S)$ (similarly, $\psi(v_b) \neq \psi(S)$). Hence, $\psi(S)$ is a $(\psi(v_a) - \psi(v_b))$-separator in $H$.
\end{proof}

\subsection{Preliminary Result Regarding Grid Contraction}

Suppose we are given a graph $G$ with a mapping $\psi$ such that $G$ is $k$-contractible to $\grid_{r \times q}$ via $\psi$.
We define a notation of \emph{partible row} in $\grid_{r \times q}$ using mapping $\psi$.
We argue that if $\grid_{r \times q}$ contains a partible row then one can \emph{un-contract} all vertices in this row to obtain a larger grid from $G$.

\begin{definition}[Partible row] \label{def:partible-row} Consider a graph $G$ which is $k$-contractible to $\grid_{r \times q}$ via mapping $\psi$. The $i_o^{th}$ row in $\grid_{r \times q}$ is said to be \emph{partible} if 
for every $j$ in $[q]$, set 
$\psi^{-1}([i_o, j])$ can be partitioned into non empty sets $U_j$ and $V_j$ which satisfy following properties:
\begin{itemize}
\item $G[U_j]$ and $G[V_j]$ are connected.
\item $U_j$ and $V_{j'}$ are adjacent if and only if $j = j'$.
\item $U_j$ and $U_{j'}$ (similarly $V_j$ and $V_{j'}$) are adjacent if and only if $|j - j'| = 1$. 
\item Let $U = \bigcup_{j \in [q]} U_j$ and $V = \bigcup_{j \in [q]} V_j$. If sets $U, V$ are adjacent with $C_b, C_f$, respectively, then sets $U, C_f$ (sets $V, C_b$) are not adjacent. 
\end{itemize}  
Here, $C_b := \{x \in V(G) |\ \psi(x) = [i, j] \text{ for some } i < i_o \text{ and } j \in [q]\}$; $C_f := \{x \in V(G) |\ \psi(x) = [i, j] \text{ for some } i > i_o \text{ and } j \in [q]\}$.
\end{definition}

\begin{lemma}\label{lemma:partible-row} Consider a graph $G$ which is $k$-contractible to $\grid_{r \times q}$ via mapping $\psi$. If $\grid_{r \times q}$ has a partible row then $G$ is $(k - q)$ contractible to $\grid_{(r + 1) \times q}$.
\end{lemma}
\begin{proof}  
Let $i_o^{th}$ row be a partible row in $\grid_{r \times q}$.
For $j$ in $[q]$, let $U_j$ and $V_j$ be the partition of $\psi^{-1}([i_o, j])$ which satisfy properties mentioned in Definition~\ref{def:partible-row}.
Also, let $U = \bigcup_{j \in [q]} U_j$ and $V = \bigcup_{j \in [q]} V_j$.
Let set $C_b$ (set $C_f$) be the collection of vertices in $G$ which are mapped to vertices in rows $\{1, 2, \dots, i_o-1\}$ (in rows $\{i_o + 1, \dots, r\}$). 
Formally, $C_b = \{x \in V(G) |\ \psi(x) = [i, j] \text{ for some } i < i_o \text{ and } j \in [q]\}$ and $C_f = \{x \in V(G) |\ \psi(x) = [i, j] \text{ for some } i > i_o \text{ and } j \in [q]\}$.
Without loss of generality, we can assume that $U, V$ are adjacent with $C_b, C_f$, respectively. 
Since set $U$ (set $V$) is a separator in $G$, sets $U, C_f$ (sets $V, C_b$) are not adjacent with each other. 
Note that $\{C_b, U, V, C_f\}$ is a partition of $V(G)$.  
We define a function $\phi : V(G) \rightarrow \grid_{(r + 1) \times q}$ on $V(G)$ as follows: for every $x \in C_b$, $\phi(x) = \psi(x)$; for every $x \in C_f$, if $\psi(s) = [i, j]$ then $\phi(x) = [i' + 1, j]$; for every $x \in U_j$, $\phi(x) = [i_o, j]$; and for every $x \in V_j$, $\phi(x) = [i_o + 1, j]$. 
Since $U_j, V_j$ are non-empty sets and $\psi$ is an onto function, $\phi$ is also an onto function. 
We argue that $\phi$ satisfy both the properties mentioned in Definition~\ref{def:graph-contractioon}.

For every vertex $[i, j]$ in $\grid_{r \times q}$, set $\psi^{-1}([i, j])$ is connected in $G$.  
Since $G[U_j], G[V_j]$ are connected, for every $[i, j]$ in $\grid_{(r + 1) \times q}$, set $\phi^{-1}([i, j])$ is connected in $G$.
This proves the first property in Definition~\ref{def:graph-contractioon}.
%Consider two vertices $[i_1, j_1]$ and $[i_2, j_2]$ in $\grid_{(r + 1) \times q}$. 
To prove the second property, we argue that any two vertices, say $[i_1, j_1]$ and $[i_2, j_2]$, in $\grid_{(r + 1) \times q}$ are adjacent with each other if and only if $\phi^{-1}([i_1, j_1])$ and $\phi^{-1}([i_2, j_2])$ are adjacent with each other.
Without loss of generality, we can assume that $i_1 \le i_2$.
Depending on the position of these two vertices in $\grid_{(r + 1) \times q}$, we consider following five cases: 
$(i)$ $i_1, i_2 < i_o$, 
$(ii)$ $i_1 < i_o$ and $i_2 \in \{i_o, i_o + 1\}$,
$(iii)$ $\{i_1, i_2\} \subseteq \{i_o, i_o + 1\}$,
$(iv)$ $i_1 \in \{i_o, i_o + 1\}$ and $i_2 > i_0 + 1$, and
$(v)$ $i_1, i_2 > i_o + 1$.

Consider Case~$(i)$. By definition of $\phi$, for $i_1, i_2 < i_o$ we have $\phi^{-1}([i_1, j_1]) = \psi^{-1}([i_1, j_1])$ and $\phi^{-1}([i_2, j_2]) = \psi^{-1}([i_2, j_2])$.
For $i_1, i_2 < i_o$, by the properties of $\psi$, there is an edge between $[i_1, j_1]$ and $[i_2, j_2]$ if and only if there is an edge between $\psi^{-1}([i_1, j_1])$ and $\psi^{-1}([i_2, j_2])$ in $G$. 
Hence we can conclude that for $i_1, i_2 < i_o$ there is an edge between $[i_1, j_1]$ and $[i_2, j_2]$ if and only if there is an edge between $\phi^{-1}([i_1, j_1])$ and $\phi^{-1}([i_2, j_2])$ in $G$. 
We can argue Case~$(v)$ by similar arguments on $\psi^{-1}([i, j])$ and $\phi^{-1}([i + 1, j])$. 
Consider Case~$(iii)$. Since $U_j, V_{j'}$ are adjacent with each other if and only if $j = j'$ and $U_j, U_{j'}$ (similarly $V_j, V_{j'}$) are adjacent with each other if and only if $|j - j'| = 1$.
Hence, the second property is satisfied.

We now argue Case~$(ii)$.
By the definition of $\phi$, we have $\phi^{-1}([i_1, j_1]) = \psi^{-1}([i_1, j_1])$ and $\phi^{-1}([i_2, j_2]) \subseteq \psi^{-1}([i_2, j_2])$.
If there is an edge between $[i_1, j_1]$ and $[i_2, j_2]$ in $\grid_{(r + 1) \times q}$ then $i_1 = i_o - 1$, $i_2 = i_o$, and $j_1 = j_2 = j$ (say).
By the property of $\psi$, there is an edge between $\psi^{-1}([i_1, j])$ and $\psi^{-1}([i_2, j])$. 
Since $C_b$ is adjacent with $U$ and non adjacent with $V$, we know that $\phi^{-1}([i_1, j])$ ($\subseteq C_b$) is adjacent with $U_{j}$ ($ \subseteq U$) and non-adjacent with $V_j$ ($\subseteq V$).
As $U_j = \phi^{-1}([i_o, j]) = \phi^{-1}([i_2, j_2])$, we can conclude that there exists an edge between $\phi^{-1}([i_1, j_1])$ and $\phi^{-1}([i_2, j_2])$. 
In reverse direction, suppose there exists an edge between $\phi^{-1}([i_1, j_1])$ and $\phi^{-1}([i_2, j_2])$. 
Since  $\phi^{-1}([i_1, j_1]) = \psi^{-1}([i_1, j_1])$ and $\phi^{-1}([i_2, j_2]) \subseteq \psi^{-1}([i_2, j_2])$, this implies there exists an edge between $\psi^{-1}([i_1, j_1])$ and $\psi^{-1}([i_2, j_2])$.
By the property of $\psi$, vertices $[i_1, j_1]$ and $[i_2, j_2]$ are adjacent with each other in $\grid_{r \times q}$.
Since $i_1 < i_o$ and $i_2 \in \{i_o, i_o + 1\}$, we can conclude that $i_2 - i_1 = 1$ and $j_1 = j_2$. This implies $[i_1, j_1], [i_2, j_2]$ are adjacent with each other in $\grid_{(r + 1)\times q}$. We can argue Case~$(iv)$ using similar arguments. 

Hence $\phi : V(G) \rightarrow \grid_{(r + 1) \times q}$ satisfies both the properties mentioned in Definition~\ref{def:graph-contractioon}.
This implies graph $G$ is contractible to $\grid_{(r + 1) \times q}$ via $\phi$. As $G$ is $k$-contractible to $\grid_{r\times q}$ and $|V(\grid_{r \times q})| + q = |V(\grid_{(r + 1)\times q})|$, we can conclude that graph $G$ is $(k - q)$-contractible to $\grid_{(r + 1) \times q}$.
\end{proof}

\subsection{Parameterized Complexity}
An instance of a parameterized problem comprises of an input $I$, which is an input of the classical instance of the problem and an integer $k$, which is called as the parameter.
A problem $\Pi$ is said to be \emph{fixed-parameter tractable} or in \FPT\ if given an instance $(I,k)$ of $\Pi$, we can decide whether or not $(I,k)$ is a \yes\ instance of $\Pi$ in  time $f(k)\cdot |I|^{\OO(1)}$.
Here, $f(\cdot)$ is some computable function whose value depends only on $k$. 
We say that two instances, $(I, k)$ and $(I', k')$, of a parameterized problem $\Pi$ are \emph{equivalent} if $(I, k) \in \Pi$ if and only if
 $(I', k') \in \Pi$.
A \emph{reduction rule}, for a parameterized problem $\Pi$ is an algorithm that takes an instance $(I, k)$ of $\Pi$ as input and outputs an instance $(I', k')$ of $\Pi$ in time polynomial in $|I|$ and $k$.
If $(I, k)$ and $(I', k')$ are equivalent instances then we say the reduction rule is \emph{safe}.
A parameterized problem $\Pi$ admits a kernel of size $g(k)$ (or $g(k)$-kernel) if there is a polynomial time algorithm (called {\em kernelization algorithm}) which takes as an input $(I,k)$, and in time $|I|^{\OO(1)}$ returns an equivalent instance $(I',k')$ of $\Pi$ such that $|I'| + k' \leq g(k)$.
Here, $g(\cdot)$ is a computable function whose value depends only on $k$.
For more details on parameterized complexity, we refer the reader to the books of Downey and Fellows~\cite{DF-new}, Flum  and Grohe~\cite{flumgrohe}, Niedermeier~\cite{niedermeier2006}, and the more recent books by Cygan et al.~\cite{saurabh-book} and Fomin et al.~\cite{fomin2019kernelization}.

\section{Combinatorial Lemma}
\label{sec:comb_lemma}

We introduce the notion of $r$-slabs which can be thought of as connected components with special properties. 
A $r$-slab is a connected set that can be partitioned into $r$ connected subsets such that the adjacency between these parts and their neighbourhood follows a certain pattern.
For an integer $r$ and a set $A$, an \emph{ordered $r$-partition} is a list of subsets of $A$ whose union is $A$. 
We define $r$-slab as follows.  
\begin{definition}[$r$-Slab]\label{def:r-slab} A \emph{$r$-slab} in $G$ is an ordered $r$-partition of a connected set $A$, say $A_1, A_2, \dots, A_r$, which satisfy following conditions. 
  \begin{itemize}
  \item For every $i$ in $[r]$, set $A_i$ is a non-empty set and $G[A_i]$ is connected.
  \item For $i \neq j$ in $[r]$, sets $A_i, A_j$ are adjacent if and only if $|i - j| = 1$.
  \item For every $i$ in $[r]$, define $B_i = N(A_i) \setminus A$. For $i \neq j$ in $[r]$, sets $B_i, B_j$ are mutually disjoint and if $B_i$ and $B_j$ are adjacent then $|i - j| = 1$.
  \end{itemize}
\end{definition}
We denote a $r$-slab by $\langle A_1, A_2, \dots, A_r \rangle$.
For a $r$-slab $\langle A_1, A_2, \dots, A_r \rangle$, set $A$ denotes union of all $A_i$s.
We note that every connected subset of $G$ is an $1$-slab.

For positive integers $\al, \bt$, a connected set $A$ in graph $G$ is called an $(\al, \bt)$-\emph{connected set} if $|A|\leq \al$ and $|N(A)|\leq \bt$.
For a non-empty set $Q\subseteq V(G)$ a connected set $A$ in $G$ is a $(Q)$-\emph{connected set} if $Q \subseteq A$. 
We generalize these notations for $r$-slab as follows.
\begin{definition}[$(\al, \bt)$-$r$-slab] For a graph $G$ and integers $\al, \bt$, a $r$-slab $\langle A_1, A_2, \dots, A_r \rangle$ is said to be an \emph{$(\al, \bt)$-$r$-slab} if $|A| \le \al$ and $|N(A)| \le \bt$.
\end{definition}
For a set $Q$, let $\calP_r(Q) = \{Q_1, Q_2, \dots, Q_r\}$ denotes its ordered $r$-partition. An ordered $r$-partition is said to be \emph{valid} if for any two vertices $u \in Q_i$ and $v \in Q_j$, $u, v$ are adjacent implies $|i - j| \le 1$.
\begin{definition}[$\calP_r(Q)$-$r$-slab]
For a graph $G$, a subset $Q$ of $V(G)$ and its ordered valid partition $\calP_r(Q) = \{Q_1, Q_2, \dots, Q_r\}$, a $r$-slab $\langle A_1, A_2, \dots, A_r \rangle$ in $G$ is said to be a \emph{$\calP_r(Q)$-$r$-slab} if $Q_i$ is a subset of $A_i$ for every $i$ in $[r]$.
\end{definition}
\begin{figure}[t]
  \begin{center}
    \includegraphics[scale=0.5]{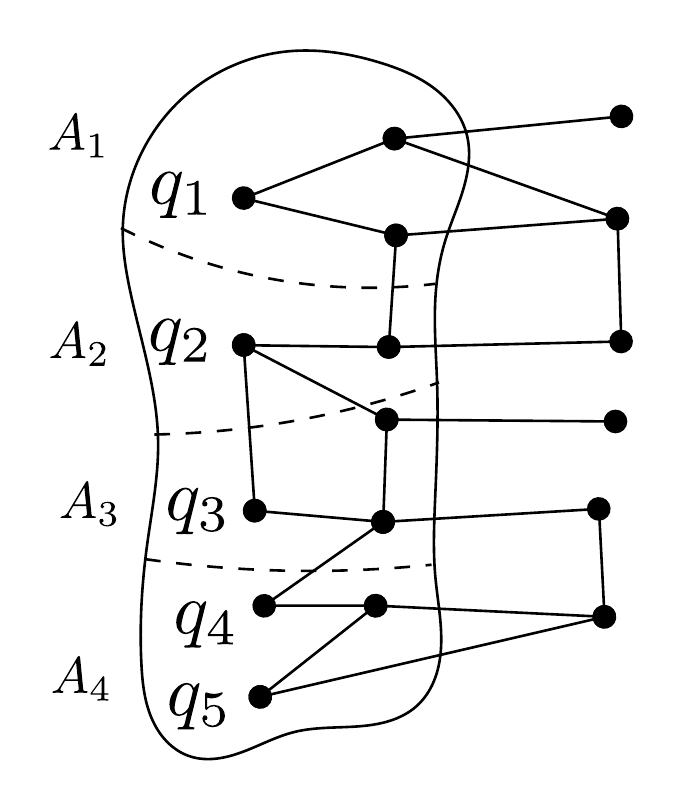}
  \end{center}
  \caption{An example of a $4$-slab. See Definition~\ref{def:r-slab}. For $Q = \{q_1, q_2, q_3, q_4, q_5\}$ and its partition $\calP_4(Q) = \{\{q_1\}, \{q_2\}, \{q_3\}, \{q_4, q_5\} \}$, $A$ is an $(\calP_4(Q), \al, \bt)$-$4$-slab.  \label{fig:slab}}
\end{figure}
See Figure~\ref{fig:slab} for an example. We combine properties mentioned in previous two definition to define specific types of $r$-slabs. 
\begin{definition}[$(\calP_r(Q),\al,\bt)$-$r$-slab] \label{def:r-slab-gen}
For a graph $G$, a non-empty subset $Q$ of $V(G)$, its ordered valid partition $\calP_r(Q) = \{Q_1, Q_2, \dots, Q_r\}$, and integers $\al, \bt$, a $r$-slab $\langle A_1, A_2, \dots, A_r \rangle$ in $G$ is a \emph{$(\calP_r(Q),\al,\bt)$-$r$-slab} if it is an $(\al, \bt)$-$r$-slab as well as a $\calP_r(Q)$-$r$-slab.
\end{definition}

We mention following two observations which are direct consequences of the definition. 
\begin{observation} \label{obs:slab-subgraph} Let $\langle A_1, A_2, \dots, A_r \rangle$ be a $(\calP_r(Q),\al,\bt)$-$r$-slab in graph $G$. If a vertex $v$ is in $N(A)$ then $\langle A_1, A_2, \dots, A_r \rangle$ is a $(\calP_r(Q),\al,\bt - 1)$-$r$-slab in graph $G - \{v\}$.
\end{observation}
For a graph $G$, consider a vertex $v$ and let $G' = G - \{v\}$. For a non-empty subset $Q'$ of $V(G')$, its ordered partition $\calP_r(Q') = \{Q'_1, Q'_2, \dots, Q'_r\}$, and integers $\al, \bt$, let $\langle A'_1, A'_2, \dots, A'_r \rangle$ be a $(\calP_r(Q'),\al,\bt)$-$r$-slab in $G'$. 
\begin{observation} \label{obs:slab-supergraph}  
If vertex $v$ satisfy following two properties then $\langle A'_1, A'_2, \dots, A'_r \rangle$ is a $(\calP_r(Q'),\al,\bt + 1)$-$r$-slab in $G$.
\begin{itemize}
\item Vertex $v$ is adjacent with exactly one part, say $A'_i$, of the $r$-slab 
\item For any vertex $u$ in $N_G'(A'_j) \setminus A'$, if $u$ and $v$ are adjacent in $G$ then $|i - j| \le 1$.
\end{itemize}
\end{observation}

Definition~\ref{def:r-slab-gen} generalizes the notation of $(Q, \al, \bt)$-\emph{connected set} defined in \cite{path-contraction}.
In the same paper, authors proved that there is an algorithm that given a graph $G$ on $n$ vertices, a non-empty set $Q \subseteq V(G)$, and integers $\al, \bt$, enumerates all $(Q, \al, \bt)$-connected sets in $G$ in time $2^{\al - |Q| + \bt} \cdot n^{\calO(1)}$.
We present similar combinatorial lemma for $(\calP_r(Q), \al, \bt)$-$r$-slabs.
\begin{lemma} \label{lemma:main-nr-r-slabs} There is an algorithm that given a graph $G$ on $n$ vertices, a non-empty set $Q \subseteq V(G)$, its ordered partition $\calP_r(Q) = \{Q_1, Q_2, \dots, Q_r\}$, and integers $\al, \bt$, enumerates all $(\calP_r(Q), \al, \bt)$-$r$-slabs in $G$ in time $4^{\al - |Q| + \bt} \cdot n^{\calO(1)}$.
\end{lemma}

\begin{proof} Let $N(Q) = \{v_1, v_2, \dots, v_p\}$. Arbitrarily fix a vertex $v_l$ in $N(Q)$.
We partition $(\calP_r(Q), \al, \bt)$-$r$-slabs in $G$ based on whether $v_l$ is contained in it or not.
In later case, such $(\calP_r(Q), \al, \bt)$-$r$-slab is also a $(\calP_r(Q), \al, \bt - 1)$-$r$-slab in $G - \{v\}$.
We now consider the first case.
Let $i$ be the smallest integer in $[r]$ such that $v_l$ is adjacent with $Q_i$. Note that, by definition, if $v_l$ is present in a $\calP_r(Q)$-$r$-slab then it can be part of either $A_{i-1}$, $A_i$ or $A_{i+1}$. We encode this fact by moving $v_l$ to either $Q_{i-1}$, $Q_i$ or $Q_{i+1}$. Let $\calP^{i-1}_r(Q \cup \{v_l\}), \calP^{i}_r(Q \cup \{v_l\})$ and $\calP^{i+1}_r(Q \cup \{v_l\})$ be $r$-partitions of $Q \cup \{v_l\}$ obtained from $\calP_r(Q)$ by adding $v_l$ to set $Q_{i-1}, Q_{i}$ and $Q_{i+1}$, respectively. Formally, these three sets are defined as follows.
\begin{itemize}
\item[-] $\calP^{i-1}_r(Q \cup \{v_l\}) := \{Q_1, \dots, Q_{i-1}\cup \{v_l\}, Q_i, Q_{i + 1}, \dots, Q_r\}$
\item[-] $\calP^{i}_r(Q \cup \{v_l\}) \ \ \ := \{Q_1, \dots, Q_{i-1}, Q_i\cup \{v_l\}, Q_{i + 1}, \dots, Q_r\}$
\item[-] $\calP^{i+1}_r(Q \cup \{v_l\}) := \{Q_1, \dots, Q_{i-1}, Q_i, Q_{i + 1}\cup \{v_l\}, \dots, Q_r\}$
\end{itemize}

\noindent \textbf{Algorithm :} We present a recursive enumeration algorithm which takes $(G, \calP_r(Q), \al, \bt)$ as an input and outputs a set, say $\calA$, of all $(\calP_r(Q), \al, \bt)$-$r$-slab in $G$. The algorithm initializes $\calA$ to an empty set. 
The algorithm returns $\calA$ if one of the following statements is true:
$(i)$ $\calP_r(Q)$ is not a valid partition of $Q$, 
$(ii)$ $\alpha - |Q| < 0$ or $\beta < 0$,
$(iii)$ there is a vertex $v_l$ in $N(Q)$ which is adjacent with $Q_{i}$ and $Q_j$ for some $i, j$ in $[r]$ such that $|i - j| \ge 2$.
If $\alpha - |Q| + \beta = 0$, the the algorithm checks if $\calP_r(Q)$ is a $(\calP_r(Q), 0, 0)$-$r$-slabs in $G$. If it is the case then the algorithm returns singleton set containing $\calP_r(Q)$ otherwise it returns an empty set.
If there is a vertex $v_l$ in $N(Q)$ which is adjacent with $Q_{i-1}, Q_i$ and $Q_{i+1}$ for some $i$ in $[r]$ then the algorithm calls itself on instance $(G, \calP_r^{i}(Q \cup \{v_l\}), \al, \bt)$ where $\calP_r^{i}(Q \cup \{v\})$ is $r$-partition as defined above. It returns the set obtained on this recursive call as the output. 
If there are no such vertices in $N(Q)$, then for some $l \in \{1, \dots, |N(Q)| \}$, the algorithm creates four instances viz $(G - \{v_l\}, \calP_r(Q), \al, \bt - 1)$ and $(G, \calP_r^{i_0}(Q \cup \{v_l\}), \al, \bt)$ for $i_0 \in \{i-1, i, i + 1\}$. The algorithm calls itself recursively on these four instances. 
Let $\calA_l^{v}, \calA_l^{i-1}, \calA_l^{i}$, and $\calA_{i+1}$ be the set returned, respectively, by the recursive call of the algorithm.
The algorithm adds all elements in $\calA_l^{i-1} \cup \calA_l^{i} \cup \calA^{i+1}_{l}$ to $\calA$.
For every $(\calP_r(Q), \al, \bt - 1)$-$r$-slabs $\langle A'_1, A'_2, \dots, A'_r \rangle$ in $\calA_l^{v}$, the algorithm checks whether it is a $(\calP_r(Q), \al, \bt)$-$r$-slabs in $G$ using Observation~\ref{obs:slab-supergraph}. If it is indeed a $(\calP_r(Q), \al, \bt)$-$r$-slabs in $G$ then it adds it to $\calA$. The algorithm returns $\calA$ at the end of this process. 

We now argue the correctness of the algorithm. 
For every input instance $(G, \calP_r(Q), \al, \bt)$ we define its measure as $\mu((G, \calP_r(Q), \al, \bt)) = \al -|Q| + \bt $.  
We proceed by the induction hypothesis that the algorithm is correct on any input whose measure is strictly less than $\al -|Q| + \bt$.  
Consider the base cases $\al - |Q| + \bt = 0$. 
In this case, the only possible $(\calP_r(Q), \al, \bt)$-$r$-slab is $\calP_r(Q)$. 
The algorithm checks this and returns the correct answer accordingly.
We consider the case when $\al - |Q| + \bt \ge 1$.  
Every $(\calP_r(Q \cup \{v_l\}), \al, \bt)$-$r$-slab is also a $(\calP_r(Q), \al, \bt)$-$r$-slab. 
The algorithm adds a $r$-slab in $\calA_l^{v}$ to $\calA$ only if it is a $(\calP_r(Q), \al, \bt)$-$r$-slabs in $G$. 
Hence the algorithm returns a set of $(\calP_r(Q), \al, \bt)$-$r$-slabs in $G$.
In remaining part we argue that every $(\calP_r(Q), \al, \bt)$-$r$-slabs is enumerated by the algorithm.

By Definition~\ref{def:r-slab}, no vertex in closed neighbhorhood of a $r$-slab can be adjacent to two non-adjacent parts of a $r$-slab.
Hence, if there is a vertex $v_l$ in $N(Q)$ which is adjacent with $Q_{i}$ and $Q_j$ for some $i, j$ in $[r]$ such that $|i - j| \ge 2$ then the algorithm correctly returns an empty set.
Suppose there exists a vertex $v$ in $N(Q)$ which is adjacent with $Q_{i-1}, Q_i$ and $Q_{i+1}$ for some $i$ in $[r]$. By Definition~\ref{def:r-slab}, any $r$-slab containing $\calP_r(Q)$ must contains $v$ in it. In this case, the number of $(\calP_r(Q), \al, \bt)$-$r$-slab is same as the number of $(\calP_r^{i}(Q \cup \{v\}), \al, \bt)$-$r$-slab where $\calP_r^{i}(Q \cup \{v\})$ is the $r$-partition of $Q \cup \{v\}$ obtained from $\calP_r(Q)$ by adding $v$ to $Q_i$. 
The measure for input instance $(G, \calP_r^{i}(Q \cup \{v\}), \al, \bt)$ is strictly smaller than $\al - |Q| + \bt$. 
Hence by induction hypothesis, the algorithm correctly computes all $(\calP_r(Q), \al, \bt)$-$r$-slab.

Consider the case when there is no vertex which is adjacent with $Q_{i-1}, Q_i$ and $Q_{i+1}$ for any $i$ in $[r]$.
Let $v_l$ be a vertex in $N(Q)$ and there is an integer $i$ in $[p]$ such that $i$ is the smallest integer, and $v_l$ is adjacent with $Q_i$. As mentioned earlier, either $v_l$ is a part of $(\calP_r(Q), \al, \bt)$-$r$-slab or not. In first case, by Definition~\ref{def:r-slab}, $v_l$ can be part of $A_{i-1}$, $A_i$ or $A_{i+1}$ in any $\calP_r(Q)$-$r$-slab. 
The measure of input instance $(G, \calP^{i_0}_r(Q \cup \{v\}), \al, \bt)$ is $\al - |Q| + \bt - 1$.
Hence by induction hypothesis, the algorithm correctly enumerates all $(\calP^{i_0}_r(Q \cup \{v\}), \al, \bt)$-$r$-slabs in $G_l$. 
Consider a $(\calP_r(Q), \al, \bt)$-$r$-slab $\langle A_1, A_2, \dots, A_r \rangle$ in $G$ which does not contain $v_l$. By Observation~\ref{obs:slab-subgraph}, $\langle A_1, A_2, \dots, A_r \rangle$ is a $(\calP_r(Q), \al, \bt - 1)$-$r$-slab in $G - \{v\}$. By induction hypothesis, the algorithm correctly computes all $(\calP_r(Q), \al, \bt - 1)$-$r$-slabs in $G - \{v\}$. Since $\langle A_1, A_2, \dots, A_r\rangle$ is a $(\calP_r(Q), \al, \bt)$-$r$-slab in $G$, vertex $v_{l}$ satisfy both the properties mentioned in Observation~\ref{obs:slab-supergraph}. Hence algorithm adds $\langle A_1, A_2, \dots, A_r\rangle$ to the set $\calA_l$. Hence, we can conclude that the algorithm correctly enumerates all $(\calP_r(Q), \al, \bt)$-$r$-slabs in $G$

Using the induction hypothesis that the algorithm correctly outputs the set of all $(\calP_r(Q), \al, \bt)$-$r$-slabs in time $4^{\al - |Q| + \bt} \cdot n^{\calO(1)}$, the running time of the algorithm follows. This concludes the proof of the lemma.
\end{proof}

We use following corollary of Lemma~\ref{lemma:main-nr-r-slabs}.

\begin{corollary}\label{cor:nr-r-slabs} There is an algorithm that given a graph $G$ on $n$ vertices and integers $\al, \bt$, enumerates all $(\al, \bt)$-$r$-slab in $G$ in time $4^{\al + \bt} \cdot n^{\calO(1)}$.
\end{corollary}

\section{An \FPT\ algorithm for \textsc{Bounded Grid Contraction}}
\label{sec:fpt-bounded-grid}

In this section, we present an \FPT\ algorithm for \textsc{Bounded
  Grid Contraction}. We formally define the problem as follows.

\vspace{0.2cm}
\defparproblem{\textsc{Bounded Grid Contraction}}{Graph $G$ and integers $k, r$}{$k, r$}{Is $G$ $k$-contractible to a grid with $r$ rows?}
\vspace{0.2cm}

We start with a definition of nice subsets mentioned in the Introduction section.
%We present an algorithm to enumerate all nice subsets in a graph.
%Later, we present an \FPT\ algorithm that builds a dynamic programming table on nice subsets.
As mentioned before, vertices of a nice subset correspond to witness sets in the first few columns of a grid-witness structure of the input graph.
Hence boundary vertices of a nice set correspond to witness sets in some column of a grid. 
Note that we are interested in the grids that have exactly $r$-rows.
Hence, we use the notation of $r$-slab defined in previous section to formally define nice sets.
%Recall that for a $r$-slab $\langle D_1, D_2, \dots, D_r \rangle$, set $D = D_1\cup D_2 \cup \cdots \cup D_r$.
Consider a $r$-slab $\langle D_1, D_2, \dots, D_r \rangle$ which corresponds to a column in some grid that can be obtained from the input graph with at most $k$ edge contraction. 
By Observation~\ref{obs:witness-structure-property}, an edge contraction reduces the number of vertices by exactly one.
As there are $3r$ many vertices in three adjacent rows in a grid, the size of a closed neighborhood of $D$ in $G$ is at most $k + 3r$.
Thus, we can focus our attention on $r$-slabs with the bounded closed neighborhood. 
We define $k$-potential $r$-slabs as follows.
\begin{definition}[$k$-Potential $r$-Slab] \label{def:k-poten-r-slab} For a given graph $G$ and integers $k, r$,  a  $r$-slab $\langle D_1, D_2, \dots, D_r \rangle$ is said to be a \emph{$k$-potential $r$-slab} of $G$ if it satisfies following two conditions:
  \begin{itemize}
  \item $|D| + |N(D)| \le k + 3r$; and
  \item $G - D$ has at most two connected components.
  \end{itemize}
  Here, $D = D_1\cup D_2 \cup \cdots \cup D_r$.
\end{definition}
\begin{definition}[Nice Subset]\label{def:nice-subset} A subset $S$ of $V(G)$ is said to be a \emph{nice subset} of $G$ if there exists a $k$-potential $r$-slab, say $\langle D_1, D_2, \dots, D_r \rangle$, such that $D$ is a subset of $S$ and $G[S \setminus D]$ is one of the connected components of $G - D$. We say that $r$-slab $\langle D_1, D_2, \dots, D_r \rangle$ is \emph{responsible} for nice subset $S$.
\end{definition}
Since $\langle D_1, D_2, \dots, D_r \rangle$ is a $k$-potential $r$-slab, both $G[S]$ and $G - S$ are connected. 
There may be more than one $k$-potential $r$-slabs responsible for a nice subset.
We define a pair of nice sets and $k$-potential $r$-slabs responsible for it.
\begin{definition}[Valid Tuple] 
A tuple $(S, \calP_r(D))$ is called a \emph{valid tuple} if $S$ is a nice subset and $\calP_r(D) \equiv \langle D_1, D_2, \dots, D_r \rangle$ is a $k$-potential $r$-slab responsible for it. 
\end{definition}

\begin{figure}[t]
  \begin{center}
    \includegraphics[scale=0.5]{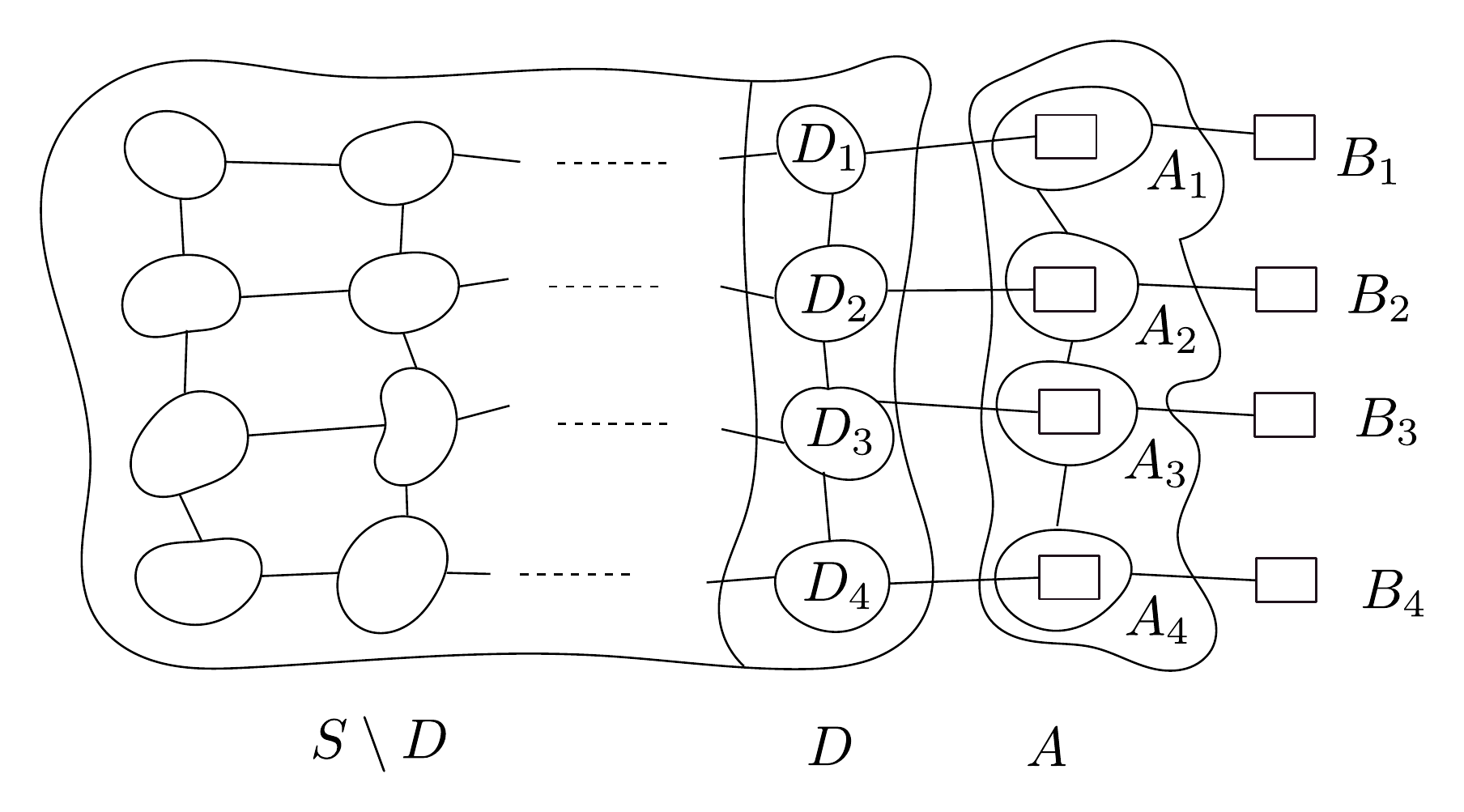}
  \end{center}
  \caption{ All sets with smooth (non-rectangular) boundary are connected. 
    Set $A$ is a possible extension of nice subset $S$. In other words, $A$ is an element in $\calA_{|A|, |B|}[(S, \calP_r(D))]$. See paragraph before Lemma~\ref{lemma:algorithm-correct}. \label{fig:column-expander}}
\end{figure}

Let $\calV_k$ be the set of all valid tuples.
For a valid tuple $(S, \calP_r(D))$ in $\calV_k$, we define a collection of $k$-potential $r$-slabs which is denoted by $\calA[(S, \calP_r(D))]$.
This set can be thought of as a collection of ``potential column extenders'' for $S$. See Figure~\ref{fig:column-expander}.
In other words, we can append a $k$-potential-$r$-slab in $\calA[(S, \calP_r(D))]$ to get a grid witness structure of a larger graphs containing $S$.
Let $\calP_{r}(A)$ be a $k$-potential-$r$-slab in $\calA[(S, \calP_r(D))]$.
Intuitively speaking, $\calP_r(A)$ is the ``new'' column to be ``appended'' to a grid witness structure of $G[S]$, to obtain a grid witness structure for $G[S \cup A]$. 
Hence if $G[S]$ can be $k'$-contracted to a grid then $G[S \cup A]$ can be $k' + (|A| - r)$-contracted to a grid. 
For improved analysis, we concentrate on subset $\calA_{a,b}[(S, \calP_r(D))]$ of $\calA[(S, \calP_r(D))]$ defined for integers $a, b$.
The set $\calA_{a,b}[(S, \calP_r(D))]$ is a collection of $k$-potential $r$-slabs of size at most $a$ which have at most $b$ neighbors outside $S$.
We impose additional condition that $a + b + |D|$ is at most $k + 3r$ for improved analysis.
Formally, $\calA_{a, b}[(S, \calP_r(D))] = \{\langle A_1, A_2, \dots, A_r \rangle \mid |A| \le a, |N(A) \setminus S| \le b, \mbox{ where } A = A_1\cup A_2 \cup \cdots \cup A_r \mbox{ and for every $D_i$ in } \calP_r(D), (N(D_i) \setminus S) \subseteq A_i, \mbox{ and } a + b + |D| \le k + 3r\}$. 

\vspace{0.3cm}

\noindent \textbf{Algorithm :}
The algorithm takes a graph $G$ on $n$ vertices and integers $k, r$ as input and outputs either \true\ or \false. 
The algorithm constructs a dynamic programming table $\Gamma$ in which there is an entry corresponding to every index $[(S, \calP_r(D)); k']$ where $(S, \calP_r(D))$ is a valid tuple in $\calV_k$ and $k'$ is an integer in $\{0\} \cup [k]$.
It initialize values corresponding to all entries to \false. \\
\noindent (\emph{for-loop Initialization}) For a tuple $(S, \calP_r(D)) \in \calV_{k}$ such that $S = D$ and $k' \ge |S| - r = |D| - r$, the algorithm sets $\Gamma[(S, \calP_r(D)); k'] = \true$. \\
\noindent (\emph{for-loop Table})
The algorithm processes indices in the table in chronologically increasing order. 
It first checks the size of $S$, then the size of $D$, followed by $k$.
Ties are broken arbitrarily. 
At table index $[(S, \calP_r(D)); k']$, if $\Gamma[(S, \calP_r(D)); k']$ is \false\ then the algorithm continues to next tuple.
If $\Gamma[(S, \calP_r(D)); k']$ is \true\ then it runs the following for-loop at this index. \\
\noindent (\emph{for-loop at Index}) The algorithm computes the set $\calA_{a, b}[(S, \calP_r(D))]$ for every pair of integers $a\ ( \ge r), b\ (\ge 0)$ which satisfy following properties $(1)$ $a + b + |D| \leq k + 3r$, $(2)$ $k' + a - r \le k$, and $(3)$ $|N(S)| \leq a$.
For every $k$-potential $r$-slab $\calP_r(A)$ in $\calA_{a, b}[(S, \calP_r(D))]$, the algorithm sets $\Gamma[(S \cup A, \calP_r(A)); k_1]$ to \true\ for every $k_1 \ge k' + (a - r)$.\\
If $\Gamma[(V(G), \calP_r(D)); k']$ is set to \true\ for some $\calP_r(D)$ and $k'$ then the algorithm returns \true\ otherwise it returns \false. This completes the description of the algorithm.
\vspace{0.2cm}

Recall that for a given connected subset $S$ of $V(G)$, $\Phi(S)$ denotes its boundary vertices i.e. set of vertices in $S$ which are adjacent with at least one vertex outside $S$.

\begin{lemma}\label{lemma:algorithm-correct} For every tuple $(S, \calP_r(D))$ in $\calV_k$ and integer $k'$ in $\{0\} \cup [k]$,
the algorithm assign $\Gamma[(S, \calP_r(D)); k'] = $ \true\ if and only if $k' + |N(S)| - r \le k$ and there is a $(r\times q)$-grid witness structure of $G[S]$, for some integer $q$, such that $\calP_r(D)$ is collection of witness sets in an end-column and $\Phi(S)$ is in $D$.
\end{lemma}
\begin{proof}
We prove the lemma by induction on $|S| + k'$ for indices $((S, \calP_r(D)); k')$ in the dynamic programming table.
For the induction hypothesis, we assume that for a positive integer $z$ the algorithm computes $\Gamma[(S, \calP_r(D)); k']$ correctly for each $(S, \calP_r(D))$ in $\calV_k$ and $k'$ in $0 \cup [k]$ for which $|S| + k' \leq z$.

Consider the base case when $|S| = |D| = r$ and $k' = 0$. Since $D \subseteq S$, we have $S = D$. This implies $\calP_r(S) = \calP_r(D)$ is a $r$-slab.
Any connected subset of a graph can be contracted to a vertex by contracting a spanning tree.
Hence, $G[S]$ can be contracted to a $(r \times 1)$-grid by contracting $|D| - r$ many edges. 
This implies that the values assigned by the algorithm in (\emph{for-loop Initialization}) are correct.
We note that once the algorithm sets a particular value to \true, it does not change it afterwards.

Assuming induction hypothesis, we now argue that the computation of $\Gamma[\cdot]$ for indices of the form $[(S_1, \calP_r(D_1)); k_1]$ where $|S_1| + k_1 = z + 1$ are correct.
Note that if $[(S_1, \calP_r(D_1)); k_1]$ is an entry in the table then $(S_1, \calP_r(D_1))$ is a valid tuple in $\calV_k$ and $k_1$ is an integer in the set $\{0\} \cup [k]$.

$(\Rightarrow)$ Assume that $G[S_1]$ is $k_1$-contractible to a $(r \times q)$-grid such that all vertices in $\Phi(S_1)$ are in an end-column $\calP_r(D_1)$ and $k_1 + |N(S_1)| - r \le k$. We argue that the algorithm sets $\Gamma[(S_1, \calP_r(D_1); k_1]$ to \true.
Let $G[S_1]$ be $k_1$-contractible to a $(r \times q)$-grid.
If $q = 1$ then $D_1 = S_1$ and in this case algorithm correctly computes $\Gamma[(S_1, \calP_r(D)); k_1]$.
Consider the case when $q \ge 2$.
Let $\calW = \{W_{ij} \mid (i, j) \in [r] \times [q]\}$ be a $(r \times q)$-grid structure of $G$ such that $\calP_r(D)$ is collection of witness sets in an end-column and $\Phi(S_1)$ is a subset of $D$. Define $\WC_j$ as union of all witness sets in column $j$. Formally, $\WC_j = \bigcup_{i=1}^{r}W_{ij}$. Hence, $\calW = \WC_1 \cup \WC_2 \cup \dots \cup  \WC_{q - 1} \cup \WC_q$ and $\calP_r(D) = \WC_q$.
 Consider set $S_0 = \WC_1 \cup \WC_2 \cup \cdots \cup \WC_{q - 1}$. Since $q \ge 2$, $S_0$ is an non-empty set.
 Let $k_0 = k_1 - (|\WC_{q}| - r)$. We argue that $[(S_0, \WC_{q - 1}); k_0]$ is an index in the table and $|S_0| + k_0 \le z$.
 As $\calW$ is a $k_1$-grid witness structure, $|\WC_{q}| - r \le k_1$ and hence $k_0$ is a non-negative integer.
 Since $G[\WC_q]$ is a connected graph, $G - \WC_{q-1}$ has exactly two connected components viz $G[\WC_1\cup \cdots \cup \WC_{q-2}]$ and the component containing $\WC_q$.
 As $\calW$ is a $k_1$-grid witness structure, $|\WC_{q - 2}| + |\WC_{q - 1}| + |\WC_q| \le k_1 + 3r \le k + 3r$ and $N(\WC_{q-1}) \subseteq \WC_{q - 2} \cup \WC_{q}$.
 (We note that $\WC_{q -2}$ may not exists but this does not change the argument. For the sake of clarity, we do not consider this as separate case.)
 Since $|\WC_{q-1}| + |N(\WC_{q-1})| \le k + 3r$ and $G - \WC_{q-1}$ has at most two connected components, $\WC_{q - 1}$ is a $k$-potential $r$-slab.
 Note that $\langle W_{1j}, W_{2j}, \dots, W_{rj} \rangle$ is the $r$-partition of $k$-potential $r$-slab $\WC_{q-1}$.
 Hence $(S_0, \WC_{q -1})$ is a tuple in $\calV_k$ and $((S_0, \WC_{q-1}); k_0)$ is an index in the table. 
 Since $\WC_q$ is not an empty set, $|S_0| + k_0 \le |S_1| - |\WC_q| + k_1 - (|\WC_q| - r) \le z + 1 + r - 2|\WC_q|$ as $|S_1| + k_1 = z + 1$. Since $|\WC_q| \ge r \ge 1$, we conclude $|S_0| + k_0 \le z$.
 Note that $\calW \setminus \{\WC_q\}$ is a $(k_1 - |W_q| + r)$-grid witness structure for $G[S_0]$. This implies that $G[S_0]$ is $k_0$-contractible to a grid with $\WC_{q-1}$ as collection of bags in an end-column and $k_0 + |N(S_0)| - r \le k_1 \le k$. Moreover, $S_0 = S_1 \setminus \WC_{q}$, $\Phi(S_0)$ is contained in $\WC_{q-1}$.
By the induction hypothesis, the algorithm has correctly set $\Gamma[(S_0, \WC_{q-1}); k_0]$ to \true.
Let $x_0 = |\WC_{q-1}|$, $a = |\WC_q|$ and $b = |\WC_q \setminus N(S_0)| = |N(S_1)| $.  We first claim that $x_0 + a + b \le k + 3r$.
Note that $|\WC_{q - 1}| + |\WC_q| \le k_1 + 2r$ and $k_1 + b \le k + r$. Hence $|\WC_{q - 1}| + |\WC_q| + b = x_0 + a + b \le k + 3r$.
At index $[(S_0, \WC_{q-1}); k_0]$, the algorithm computes $\calA_{a,b}[(S_0, \WC_{q-1})]$.
Clearly, $\WC_q$ is one of the sets in $\calA_{a,b}[(S_0, \WC_{q-1})]$ as for every $i$ in $[r]$, $N(W_{i,q-1}) \setminus S_0$ is contained in $W_{iq}$ and $G[W_{iq}]$ is a connected graph. Hence the algorithm sets $\Gamma[(S_1, \WC_q), k_1] = \Gamma[(S_1, \calP_r(D)), k_1]$ to \true.

$(\Leftarrow)$ To prove other direction, we assume that the algorithm sets $\Gamma[(S_1, \calP_r(A)); k_1]$ to $\true$.
We argue that $G[S_1]$ is $k_1$-contractible to a grid such that $\calP_r(A)$ is a collection of witness sets in an end-column in a witness structure; $\Phi(S_1)$ is in $A$; and $k_1 + |N(S_1)| - r \le k$.
If $\Gamma[(S_1, \calP_r(A)); k_1]$ is set to $\true$ in the (\emph{for-loop Initialization}) then, as discussed in first paragraph, this is correct.
Consider the case when the value at $\Gamma[(S_1, \calP_r(A)); k_1]$ is set to \true\ when the algorithm was processing at index $[(S_0, \calP_r(D)); k_0]$.
Note that value at $[(S_0, \calP_r(D)); k_0]$ has been set \true\ by the algorithm as otherwise, it will not change any value while processing this index.
Note that $|A| = a$ and $|N(S_1)| = b$.
Since $a$ is a positive integer and $k_0 + a - r \le k_1$ (because (\emph{for-loop at Index}) updates only for such values), we know $|S_0| + k_0 \le |S_1|  + k_1 - 2a + r = z + 1 - 2a + r$.
Since $a \ge r \ge 1$, we get $|S_0| + k_0 \le z$.
By the induction hypothesis, algorithm has correctly computed value at $[(S_0, \calP_r(D)); k_0]$.
Hence $G[S_0]$ can be $k_0$-contracted to a grid
such that $\Phi(S_0)$ is in $D$ and there exists a grid witness 
structure, say $\calW_0$, such that $\calP_r(D)$ is a collection of witness sets in an end-column. The induction hypothesis also implies and $k_0 + |N(S_0)| - r \le k + 3r$.

Let $\calP_r(A) = \langle A_1, A_2, \dots, A_r \rangle$ be the $r$-partition of $A$ in $\calA_{a, b}[(S_0, \calP_r(A))]$ at which \emph{for-loop at Index} changes the value at $\Gamma[(S_1, \calP_r(A)); k_1]$. 
By construction, every $D_i$ in $\calP_r(D)$, $D_i$ is contained in $A_i$.
Since $\Phi(S_0)$ is contained in $D$, no vertex in $S_0 \setminus D$ is adjacent with any vertex in $A$.
Since $\calP_r(A)$ is a $r$-slab, $\calW_0 \cup \{A_1, A_2, \dots, A_r\}$ is a grid witness structure of $G[S_1]$.
Moreover, since $N(S_0)$ is in $A$, $\Phi(S_1)$ is contained in $A$. 
Hence, $G[S_1]$ can be $k_1$-contractible to a grid with all vertices in $\Phi(S)$ in a $A$ and there exists a witness structure for which $\calP_r(A)$ is a collection of witness sets in an end-columns.
It remains to argue that $k_1 + |N(S_1)| - r \le k$.
We prove this for the case $k_1 = k_0 + a - r$ as $k_1 > k_0 + a - r$ case follows from the definition of $k_1$-contratibility.
Let $x_0 = |D|$.  As $x_0$ is the size of an end-column in $\calW_0$, we have $x_0 - r \le k_0$.
As algorithm only considers $a, b$ such that $x_0 + a + b \le k + 3r$, substituting $a = k_1 - k_0 + r$ and $b = |N(S_1)|$ we get $x_0 + k_1 - k_0 + r + |N(S_1)| \le k + 3r$.  
Using $x_0 - r \le k_0$, we get the desired bound.

 This completes the proof of the lemma.  
\end{proof}

\begin{lemma}\label{lemma:algorithm-time}
Given a graph $G$ on $n$ vertices and integers $k, r$, the algorithm terminates in time $4^{k + 3r} \cdot n^{\calO(1)}$.
\end{lemma}
\begin{proof} We first describe an algorithm that given a graph $G$ on $n$ vertices and integers $k, r$, enumerates all valid tuples in time $4^{k + 3r} \cdot n^{\calO(1)}$. The algorithm computes all $r$-slabs in $G$ which satisfy first property in Definition~\ref{def:k-poten-r-slab} using Corollary~\ref{cor:nr-r-slabs}. 
For every $r$-slabs, it checks whether it satisfy the second property in Definition~\ref{def:k-poten-r-slab} to determine whether it is a $k$-potential $r$-slab or not.
For a $k$-potential $r$-slab $\calP_r(D) \equiv \langle D_1, D_2, \dots, D_r \rangle$, if $G - D$ has exactly one connected component, say $C_1$, the it adds $(V(C_1) \cup D, \calP_r(D))$ and $(D, \calP_r(D))$ to set of valid tuples. 
If $G - D$ has two connected components, say $C_1, C_2$, then it adds $(V(C_1) \cup D, \calP_r(D))$ and $(V(C_2) \cup D, \calP_r(D))$ to the set of valid tuples. This completes the description of the algorithm.
Note that the algorithm returns a set of valid tuples. 
For a $k$-potential $r$-slab $\calP_r(D) \equiv \langle D_1, D_2, \dots, D_r \rangle$, $G - D$ has at most two connected components.
Hence any $k$-potential $r$-slab is responsible for at most two nice subsets. 
By definition of nice subsets, for any nice subset there exists a $k$-potential $r$-slab responsible for it. Hence the algorithm constructs the set of all valid tuples.
The algorithm spends polynomial time for each $r$-slab it constructs. Hence, the running time of the algorithm follows from Corollary~\ref{cor:nr-r-slabs}.

The algorithm can computes the table and completes \emph{for-loop Initialization} in time $4^{k + 3r} \cdot n^{\calO(1)}$ using the algorithm mentioned in above paragraph.
We now argue that the \emph{for-loop Table} takes $4^{k + 3r} \cdot n^{\calO(1)}$ time to complete.
We partition the set of valid tuples $\calV_k$ using the sizes of the neighborhood of connected component and size of $r$-slab in a tuple.
For two fixed integers $x, y$, define $\calV_k^{x, y} := \{(S, \calP_r(D)) \in \calV_k |\ |D| \le x \text{ and } |N(S)| \le y \}$.
In other words, $\calV_k^{x, y}$ collection of all nice subsets whose neighborhood is of size $y$ and there is a $k$-potential $r$-slab of size $x$ responsible for it.
Alternatively, $\calV_k^{x, y}$ is a collection of $k$-nice subsets for which there is a $(x, y)$-$r$-slab is responsible for it.
Since the number of $(x, y)$-$r$-slabs are bounded (Corollary~\ref{cor:nr-r-slabs}) and each $k$-potential $r$-slab is responsible for at most two nice subsets, $|\calV_k^{x,y}|$ is bounded by $4^{x+y} \cdot n^{\calO(1)}$.

For each $(S, \calP_r(D)) \in \calV_k^{x,y}$, the algorithm considers every pair of integers $a (> 0), b (\ge 0)$, such that $x + a + b \leq k + 3$ and $ |N(S)| = y \leq a$, and computes the set $\calA_{a, b}[(S, \calP_r(D))]$.
By Lemma~\ref{lemma:main-nr-r-slabs}, set $\calA_{a, b}[(S, \calP_r(D))]$ can be computed in time $4^{a+b-|N(S)|} \cdot n^{\calO(1)}$.   The algorithm spends time proportional to $|\calA_{a, b}[(S, \calP_r(D))]|$ for \emph{for-loop at Index}.
Hence for two fixed integers $x, y$, algorithm spends 
$$
\sum_{\substack{a,b \\ x+a+b \leq k + 3r}} 
 4^{x+y} \cdot 4^{a+b-y} \cdot n^{\calO(1)} =  
\sum_{\substack{a,b \\ x+a+b \leq k + 3r}} 
 4^{x+a+b} 
 \cdot n^{\calO(1)} = 4^{k + 3r} \cdot n^{\calO(1)}
 $$
time to process all valid tuples in $\calV_k^{x, y}$.  Since there are at most $\calO(k^2)$ feasible values for $x, y$, the overall running time of algorithm is bounded by $4^{k + 3r} \cdot n^{\calO(1)}$. This concludes the proof.
\end{proof}

The following theorem is implied by Lemmas~\ref{lemma:algorithm-correct}, \ref{lemma:algorithm-time}, and the fact that $(V(G), \calP_r(D))$ is a tuple in $\calV_k$ for some $D$. 

\begin{theorem}\label{lemma:bounded-grid-contraction}
\sloppy There exists an algorithm which given an instance $(G, k, r)$ of \textsc{Bounded Grid Contraction} runs in time $4^{k + 3r} \cdot n^{\calO(1)}$ and correctly determines whether it is a \yes\ instance or not. Here, $n$ is the number of vertices in $G$.
\end{theorem}

\section{An \FPT\ algorithm for \textsc{Grid Contraction}}
\label{sec:fpt-grid}

In this section, we present an \FPT\ algorithm for \textsc{Grid Contraction}.
Given instance $(G, k)$ of \textsc{Grid Contraction} is a \yes\ instance if and only if $(G, k, r)$ is a \yes\ instance of \textsc{Bounded Grid Contraction} for some $r$ in $\{1, 2, \dots, |V(G)|\}$.
For $r < 2k + 5$, we can use algorithm presented in
Section~\ref{sec:fpt-bounded-grid} to check whether given graph can be contracted to grid with $r$ rows or not in \FPT\ time.
A choice of this threshold will be clear in the latter part of this section.
If algorithm returns \yes\ then we can conclude that $(G, k)$ is a
\yes\ instance of \textsc{Grid Contraction}.
If not then we can correctly conclude that if $G$ is $k$-contractible to a grid
then the resulting grid has at least $2k + 5$ rows.
This information allows us to find two \emph{rows} in $G$ which can safely be contracted.
We need the following generalized version of \textsc{Grid Contraction} to state these results formally.

\vspace{0.3cm}
\defparproblem{\textsc{Annotated Bounded Grid Contraction}}{Graph $G$, integers $k, r, q$, and a tuple $(x_1, x_2, x_3, x_4)$ of four different vertices in $V(G)$}{$k, r$}{Is $G$ $k$-contractible to $\grid_{r \times q}$ such that there is a $\grid_{r \times q}$-witness structure of $G$ in which the witness sets containing $x_1, x_2, x_3$, and $x_4$ correspond to four corners in $\grid_{r \times q}$?}
\vspace{0.3cm}

Assume that $G$ is $k$-contractible to $\grid_{r\times q}$ with desired properties via mapping $\psi$.
Let $t_1, t_2, t_3,$ and $t_4$ be corners in $\grid_{r \times q}$ such that  $t_1 \equiv [1, 1], t_2 \equiv [1, q], t_3 \equiv [r, q],$ and $t_4 \equiv [r, 1]$.
There are $4!$ ways in which vertices in $\{x_1, x_2, x_3, x_4\}$ can be uniquely mapped to corners $\{t_1, t_2, t_3, t_4\}$.
For the sake of simplicity, we assume that we are only interest in the case in which $x_1, x_2, x_3, x_4$ are mapped to $t_1, t_2, t_3,$ and $t_4$ respectively.
In other words, $\psi(x_i) = t_i$ for all $i \in \{1, 2, 3, 4\}$.

We can modify the algorithm presented in
Section~\ref{sec:fpt-bounded-grid} obtain an algorithm for \textsc{Annotated Bounded Grid Contraction} problem which is fixed parameter tractable when parameterized by $(k + r)$.
The modified algorithm only initializes tuple $(S, \calP_r(D)) \in \calV_{k}$ such that $S = D$,  $k' \ge |S| - r = |D| - r$, and $x_1, x_2$ are in first and last parts in $\calP_r(D)$ in the (\emph{for-loop Initialization}) step.
Recall that the algorithm in Section~\ref{sec:fpt-bounded-grid} set  $\Gamma[(S, \calP_r(D)); k']
= $ \true\ if and only if $k' + |N(S)| - r \le k$ and there is a
$(r \times q')$-grid witness structure of $G[S]$, for some integer $q'$, such that $\calP_r(D)$ is collection of witness sets in an end-column and $\Phi(S)$ is in $D$.
Instead of storing \true\ or \false, the modified algorithm  
stores $q'$ if it is \true\ and $0$ otherwise.
With these simple modifications, we obtain the following result.

\begin{lemma}\label{lemma:algo-annotated-bounded-grid} 
There exists an algorithm which given an instance $(G, k, r, q, (x_1,
x_2, x_3, x_4))$ of \textsc{Annotated Bounded Grid Contraction} runs
in time $4^{k + 3r} \cdot n^{\calO(1)}$ and correctly determines
whether it is a \yes\ instance or not. Here, $n$ is the number of
vertices in $G$.
\end{lemma}

In the case, when $r < 2k + 5$ the algorithm mentioned in the above lemma is fixed parameter tractable when the parameter is $k$ alone.
When $r \ge 2k + 5$, we argue that if $(G, k, r, q, (x_1, x_2, x_3, x_4))$ is a \yes\ instance then there exists a \emph{horizontal decomposition} of $G$ (Lemma~\ref{lemma:existence-hor-decomp}). 
We formally define horizontal decomposition as follows.

\begin{definition}[Horizontally-Decomposable]\label{def:horizontally-decomposable} Consider an instance $(G, k, r, q, (x_1, x_2, x_3, x_4))$ of \textsc{Annotated Bounded Grid Contraction}. A graph $G$ is said to be \emph{horizontally-decomposable} if $V(G)$ can be partitioned into four non-empty parts $C_{12}, S_u, S_v,$ and $C_{34}$ which satisfies following properties.
\begin{itemize}
\item The graphs $G[C_{12}], G[C_{34}]$ are connected and $x_1, x_2 \in C_{12}$, $x_3, x_4 \in C_{34}$.
\item The graph $G[S_u \cup S_v]$ is a $2 \times q$ grid with $S_u, S_v$ correspond to vertices in its two rows.
\item $C_{12}$ and $C_{34}$ are the two connected components of $G \setminus (S_u \cup S_v)$.
\item $N(C_{12}) = S_u$ and $N(C_{34}) = S_v$.
\end{itemize}
\end{definition}

\begin{lemma}\label{lemma:existence-hor-decomp} Consider an instance $(G, k, r, q, (x_1, x_2, x_3, x_4))$ of \textsc{Annotated Bounded Grid Contraction} such that $2k + 5 \le r$. If it is a \yes\ instance then there exists a horizontal decomposition of $G$.
\end{lemma}
\begin{proof} 
Assume that  $G$ is $k$-contractible to $\grid_{r \times q}$ with desired properties via mapping $\psi$.
By Observation~\ref{obs:witness-structure-property}, there are at most $k$ big-witness sets.
This implies that there are at most $k$ rows in $\grid_{r \times q}$ which contain vertices corresponding to big-witness sets.
Since there are at least $2k + 5$ rows in $\grid_{r \times q}$, there exists $i_o$ in $\{2, 3, \dots, r - 2\}$ such that no vertex in $i_o^{th}$ and $(i_o + 1)^{th}$ row corresponds to a big witness set. Define $C_{12}, S_u, S_v, C_{34}$ as follows.
\begin{itemize}
\item $C_{12} := \{x \in V(G) |\ \psi(x) = [i, j] \text{ for some } i < i_o \text{ and } j \in [q]\}$
\item $S_{u} := \{x \in V(G) |\ \psi(x) = [i_o, j] \text{ for some } j \in [q]\}$.
\item $S_{v} := \{x \in V(G) |\ \psi(x) = [i_o + 1, j] \text{ for some }  j \in [q]\}$
\item $C_{34} := \{x \in V(G) |\ \psi(x) = [i, j] \text{ for some } i > i_o + 1 \text{ and } j \in [q]\}$
\end{itemize}
It is easy to verify that $(C_{12}, S_u, S_v, C_{34})$ is a horizontal decomposition of $G$.
\end{proof}

Consider an instance $(G, k, r, q, (x_1, x_2, x_3, x_4))$, let $(C_{12}, S_u, S_v, C_{34})$ be a horizontal decomposition of $G$.
Reduction Rule~\ref{rr:reduce-instance} contracts all the edges across $S_u, S_v$.
Note that in the resulting instance, $r$ is decreased by one.
 
\begin{reduction rule}\label{rr:reduce-instance} For an instance $(G, k, r, q, (x_1, x_2, x_3, x_4))$, let $(C_{12}, S_u, S_v, C_{34})$ be a horizontal decomposition of $G$. Let $S_u (= \{u_1, u_2, \dots, u_{q}\})$ and $S_v (= \{v_1, v_2, \dots, v_{q}\})$.
Let $G'$ be the graph  obtained from $G$ by contracting all the edges in $\{u_jv_j |\ j \in [q]\}$.
Return instance $(G', k, r - 1, q, (x_1, x_2, x_3, x_4))$.
\end{reduction rule}

As $S_u, S_v$ are $\{(x_1-x_4), (x_2-x_3)\}$-separators in $G$, by Observation~\ref{obs:witness-structure-property}, sets $\psi(S_u), \psi(S_v)$ are  $\{(t_1-t_4), (t_2-t_3)\}$-separators in $\grid_{r\times q}$.
We argue that $\psi(S_u)$ and $\psi(S_v)$ correspond to two consecutive rows and it was safe to contract edges across $S_u, S_v$.

\begin{lemma}\label{lemma:rr-reduce-instance-safe} Reduction Rule~\ref{rr:reduce-instance} is safe.
\end{lemma}
\begin{proof} 
\noindent 
Note that $\grid_{(r - 1) \times q}$ can be obtained from $\grid_{r \times q}$ by contracting all edges across any two consecutive rows. Also, this operation does not remove any vertex from witness sets corresponding to corner vertices in grid.

$(\Rightarrow)$
Assume $G$ is $k$-contractible to $\grid_{r \times q}$ via mapping $\psi$ with desired properties.
We argue that $G'$ is $k$-contractible $\grid_{(r - 1) \times q}$ with desired properties.
Let $t_1, t_2, t_3,$ and $t_4$ be the four corners of $\grid_{r\times q}$ such that $\psi(x_i) = t_i$ for all $i \in \{1, 2, 3, 4\}$.
By Observation~\ref{obs:witness-structure-property}, $\psi(S_u), \psi(S_v)$ are $\{(t_1, t_4), (t_2, t_3)\}$-separators in $\grid_{r \times q}$. 
By Observation~\ref{obs:corner-sep-size}, $|\psi(S_u)|, |\psi(S_v)| \ge q$.
By the property of mapping $\psi$, we have $|\psi(S_u)| \le |S_u|$ and $|\psi(S_v)| \le |S_v|$. 
Since $|S_u| = |S_v| = q$, we have $|\psi(S_u)| = |\psi(S_v)| = q$.
Hence, by Observation~\ref{obs:corner-sep-size}, $\psi(S_u)$ and $\psi(S_v)$ corresponding to rows in $\grid_{r\times q}$.
Let $\psi(S_u)$ and $\psi(S_v)$ correspond to rows $i_1, i_2$. 
Since there are multiple edges across $S_u, S_v$, we have $|i_1 - i_2| \le 1$. 
As $G$ is $k$-contractible to $\grid_{r \times q}$,  if $|i_1 - i_2| = 1$ then $G'$ is $k$-contractible to $\grid_{(r - 1) \times q}$. 
If $|i_1 - i_2| = 0$ then $G'$ is $(k - q)$-contractible to $\grid_{r \times q}$ as $G$ is $k$-contractible to $\grid_{r \times q}$ and $G' = G/\{v_ju_j |\ j \in [q]\}$. 
Hence, $G'$ is $k$-contractible to $\grid_{(r - 1) \times q}$.

$(\Leftarrow)$ Let $S_o = \{s_1, s_2, \dots, s_q\}$ be the set vertices in $G'$ which are obtained by contracting edges $u_jv_j$ in $G$. 
In other words, for $j$ in $[q]$, let $s_j$ be the new vertex added while contracting edge $u_jv_j$.
Note that $S_o$ is $\{(x_1-x_4), (x_2-x_3)\}$-separator in $G'$. 

Assume that $G'$ is $k$-contractible to $\grid_{(r - 1) \times q}$ with the desired properties via mapping $\phi$.
Let $t_1, t_2, t_3,$ and $t_4$ be the four corners of $\grid_{(r - 1)\times q}$ such that $\phi(x_i) = t_i$ for $i \in \{1, 2, 3, 4\}$.
By Observation~\ref{obs:witness-structure-property}, $\phi(S_o)$ is a $\{(t_1, t_4), (t_2, t_3)\}$-separators in $\grid_{(r - 1) \times q}$. 
By Observation~\ref{obs:corner-sep-size}, $|\phi(S_o)| \ge q$.
By the property of mapping $\phi$, we have $|\phi(S_o)| \le |S_o|$. 
Since $|S_o| = q$, we have $|\phi(S_o)| = q$.
Hence, by Observation~\ref{obs:corner-sep-size}, $\phi(S_o)$ corresponding to a row, say $r_o$, in $\grid_{(r-1)\times q}$. 

As $G$ is $q$-contractible to $G'$ (which is $k$-contractible to $\grid_{(r-1)\times q}$), we know that $G$ is $(k + q)$-contractible to 
$\grid_{(r-1)\times q}$. 
A mapping $\psi : V(G) \rightarrow \grid_{r \times q}$ as follows corresponds to this contraction.
For every $x$ in $V(G) \setminus (S_u \cup S_v)$, define $\psi(x) = \phi(x)$ and for every $x$ in $\{ u_j, v_j\}$, define $\psi(x) = \phi(s_i)$.
%Note that $G$ is $(k + q)$-contractible to $\grid_{(r-1)\times q}$ via function $\psi$ with desirable properties.
We argue that $r_o^{th}$ row in $\grid_{(r-1)\times q}$ is partible.
For every $j$ in $[q]$, we define a partition $U_j, V_j$ of $\psi^{-1}([i_1, j])$ which satisfy all the properties mentioned in Definition~\ref{def:partible-row}.

For $j$ in $[q]$, let $X_j = \psi^{-1}([i_o, j])$.
As $s_j$ was present in $\phi^{-1}([i_o, j])$, vertices $u_j, v_j$ are present in $X_j$.
By the property of $\phi$, set $\phi^{-1}([i_o, j])$ is connected in $G'$. Since vertex $s_j$ is obtained from contracting edge $u_jv_j$ in $G$, graph $G[X_j]$ is connected.
Since $\phi(S_o)$ corresponds to a row of with $q$ vertices and $|S_o| = q$, vertices $u_{j'}, v_{j'}$ are present in $X_j$ if and only if $j' = j$.
In other words, $X_j \cap (S_u \cup S_v) = \{u_j, v_j\}$.
Define $U_j := (X_j \cap C_{12}) \cup \{u_j\}$ and $V_j := (X_j \cap C_{34}) \cup \{v_j\}$.
Since $(C_{12}, S_u, S_v, C_{34})$ is a horizontal decomposition of $G$, sets $U_j, V_j$ is a non-empty partition of $X_j$.
We argue that $U_j, V_j$ satisfy all the properties in Definition~\ref{def:partible-row}.
As $N(C_{12}) = S_u$, no vertex in $U_j \setminus \{u_j\}$ is adjacent with $S_v \cup C_{34}$.
By similar arguments, no vertex in $V_j \setminus \{v_j\}$ is adjacent with $S_v \cup C_{34}$. 
As $G[X_j]$ is connected and $U_j \subseteq S_u \cup C_{12}; V_j \subseteq S_v \cup C_{34}$, graphs $G[U_j], G[V_j]$ are connected.
Moreover, for $j' \in [q]$, sets $U_j, V_{j'}$ are adjacent if and only if $u_j, v_{j'}$ are adjacent.
Since $G[S_1\cup S_2]$ is a $(2 \times q)$-grid, $u_j$ and $v_{j'}$ are adjacent if and only if $j = j'$.
Hence $U_j$ and $V_{j'}$ are adjacent if and only if $j = j'$.
Since $U_j \subseteq X_j$ and $U_{j'} \subseteq X_{j'}$, $U_j$ and $U_{j'}$ are non adjacent if $|j - j'| > 1$.
If $|j - j'| = 1$ then $U_j, U_{j'}$ are adjacent as they contain $u_j$ and $u_{j'}$.
Hence $U_j, U_{j'}$ are adjacent if and only $|j - j'| = 1$.
By similar arguments, $V_j, V_{j'}$ are adjacent if and only if $|j - j'| = 1$.
As no vertex in $U_j$ is adjacent with $C_{34}$ and no vertex in $V_j$ is adjacent with $C_{12}$, we can conclude that partition $U_j, V_j$ satisfy all the properties in Definition~\ref{def:partible-row}.
 
Since $r_o^{th}$ row in $\grid_{(r-1)\times q}$ is partible, Lemma~\ref{lemma:partible-row} implies that $G$ is $k$-contractible to $\grid_{r \times q}$. This concludes the proof of reverse direction.

Hence $(G, k, r, q, (x_1, x_2, x_3, x_4))$ is a \yes\ instance of \textsc{Annotated Bounded Grid Contraction} if and only if  
$(G', k, r - 1, q, (x_1, x_2, x_3, x_4))$ is a \yes\ instance.
\end{proof}

It remains to argue that Reduction Rule~\ref{rr:reduce-instance} can be implemented in polynomial time.
In Lemma~\ref{lemma:find-hor-decomp}, we argue there exists an algorithm that can find a horizontal decomposition, if exists, in polynomial time. 
We use the following structural lemma to prove the previous statement.

\begin{lemma}\label{lemma:nr-grid-sep} Given two adjacent vertices $u_1, v_1$ in $G$, there is at most one subset $S$ of $V(G)$ such that $(a)$ $G[S]$ is a $(2 \times q)$ grid, $(b)$ $u_1, v_1$ are two vertices in the first column of $G[S]$, and $(c)$ each row in $S$ is a separator in $G$. Moreover, if such a subset exists then it can be found in polynomial time.
\end{lemma}
\begin{proof}
For the sake of a contradiction, assume that there are two such subsets, say $S, S'$, of $V(G)$.
Let $S_u, S_v$ and $S'_u, S'_v$ be the two rows in $G[S]$ and $G[S']$, respectively.
Let $j$ be the first column in which vertices in $S, S'$ differs.
As the first columns in $S, S'$ are same, $j \ge 2$.
Let $u_j, v_j$ and $u'_j, v'_j$ be the vertices in $j^{th}$ row of $S$ and $S'$ respectively.
Without loss of generality, assume $u_j \neq u'_j$.
Since $u_j, u'_j$ both are adjacent with $u_{j - 1}$ and $u_j \in S_u$, we can conclude $u'_j \not\in S_u$.
The only vertex in $S_v$ which is adjacent with $u_{j - 1}$ is $v_{j-1}$ and $u'_j \neq v_{j-1}$, we have $u'_j \not\in S_v$.
By similar argument, we can prove that $v'_j \not\in S_u$.
As $v_{j - 1}$ is in $S_v$, we have $v_{j - 1} \not\in S_u$.
To summarize, we can conclude that $u'_j \not\in S_v$ and $u'_j, v_{j-1}, v'_j \not\in S_u$.

Consider separators $S_u, S_v$ in graph $G$.
Let $C$ be a connected component of $G - S_u$ which contains $S_v$.
Since $S_v$ is also a separator, the only vertices in $C$ which are adjacent with $S_u$ are in $S_v$.
This implies that vertex $u'_j$ which is adjacent with $S_u$ can not be in $C \setminus S_v$.
Since $u'_j \not\in S_v$, we can conclude that $u'_j$ and $S_v$ are in different connected components of $G - S_u$.
This implies  $u'_j$ and $v_{j-1}$ are in different connected component of $G - S_u$. 
But, $u'_j, v_{j-1}, v'_j \not\in S_u$ and there exists a path $(u'_j, v'_j, v_{j-1})$ in $G$.
This leads to a contraction to the fact that $u'_j$ and $v_{j-1}$ are in different connected components of $G - S_u$.
Hence our assumption is wrong and there exists at most one such set.

Given $u_1, v_1$ and the uniqueness of a subset with the desired property, if there exists such subgraph then there are unique choices for vertices in the second column.
In other words,  a subset $S$ with desired properties exists if and only if there is a unique pair of adjacent vertices, say $u_2, v_2$, in graph $G$ which satisfy following conditions -- $(1)$ $u_1u_2, v_1v_2 \in E(G)$ and $(2)$ $u_1v_2, v_1u_2 \not\in E(G)$.
One can stepwise add new columns in $S$ while checking $G[S]$ remains a grid with two rows until no more columns can be added.
This algorithm terminates in polynomial time and either returns a subset with desired properties or correctly concludes that no such subgraph exists.
\end{proof}

\begin{lemma}\label{lemma:find-hor-decomp} There exists an algorithm which given an instance $(G, k, r, q, (x_1, x_2, x_3, x_4))$ of \textsc{Annotated Bounded Grid Contraction} runs in polynomial time and either returns a horizontal decomposition of $G$ or correctly concludes that no such decomposition exits.
\end{lemma}
\begin{proof}
For every pair of adjacent vertices $u_1, v_1$ in $G$, the algorithm tries to find a subset of $V(G)$ with the properties mentioned in the statement of Lemma~\ref{lemma:nr-grid-sep}.
If such a subset exists, say $S$, then the algorithm checks whether $S$ and connected components of $G - S$ satisfy the conditions mentioned Definition~\ref{def:horizontally-decomposable}.
The algorithm returns a horizontal decomposition if it finds one.
As the algorithm exhaustively searches for all possible $(2 \times q)$-grids which are also separators, if it does not return a horizontal decomposition then the graph does not admit a horizontal decomposition.
The running time of the algorithm is implied by the fact that algorithm runs over all edges in the input graph, conditions in Definition~\ref{def:horizontally-decomposable} can be checked in polynomial time, and by Lemma~\ref{lemma:nr-grid-sep}.
\end{proof}

We are now in a position to present main result of this section.
\begin{theorem}\label{theorem:grid-contraction}
\sloppy There exists an algorithm which given an instance $(G, k)$ of \textsc{Grid Contraction} runs in time $4^{6k} \cdot n^{\calO(1)}$ and correctly determines whether it is a \yes\ instance or not. Here, $n$ is the number of vertices in $G$.
\end{theorem}
\begin{proof}
The algorithm starts with checking whether graph $G$ is $k$-contractible to a path using the algorithm in \cite{tree-contraction}.
If it is then the algorithm returns \yes\ else it creates polynomially many instances of \textsc{Annotated Bounded Grid Contraction} by guessing all possible values of $r, q, x_1, x_2, x_3, x_4$. 
It processes these instances with increasing values of $r$.
Ties are broken arbitrarily.
For $r < 2k + 5$, the algorithm check whether $(G, k, r, q, (x_1, x_2, x_3, x_4))$ is a \yes\ instance of \textsc{Annotated Bounded Grid Contraction} using Lemma~\ref{lemma:algo-annotated-bounded-grid}. 
For $r \ge 2k + 5$, the algorithm checks whether there exists a horizontal decomposition of $G$ using Lemma~\ref{lemma:find-hor-decomp}. 
If there exists a horizontal decomposition of $G$ then the algorithm applies Reduction Rule~\ref{rr:reduce-instance} to obtain another instance of \textsc{Annotated Bounded Grid Contraction} with a smaller value of $r$. 
The algorithm repeats the above step until $r < 2k + 5$ or the graph in a reduced instance does not have a horizontal decomposition. 
In the first case, it checks whether a reduced instance is a \yes\ instance or not using Lemma~\ref{lemma:algo-annotated-bounded-grid}. 
In the second case, it continues to the next instance created at the start of the algorithm. 
The algorithm returns \yes\ if at least one of the instances of \textsc{Annotated Bounded Grid Contraction} is a \yes\ instance.

It is easy to see that an instance $(G, k)$ of \textsc{Grid Contraction} is a \yes\ instance if and only if there exists integers $r, q$ in $\{1, 2, \dots, |V(G)|\}$ and four vertices $x_1, x_2, x_3, x_4$ in $V(G)$ such that $(G, k, r, q, (x_1, x_2, x_3, x_4))$ is a \yes\ instance of \textsc{Annotated Bounded Grid Contraction}.
Lemma~\ref{lemma:rr-reduce-instance-safe} implies the correctness of the step where the algorithm repeatedly applies Reduction Rule~\ref{rr:reduce-instance} and check whether the reduced instance is a \yes\ instance of \textsc{Annotated Bounded Grid Contraction} or not.
Consider an instance $(G, k, r, q, (x_1, x_2, x_3, x_4))$ such that $r > 2k +5$ and there is no horizontal decomposition of $G$.
By Lemma~\ref{lemma:existence-hor-decomp}, the algorithm correctly concludes that it is a \no\ instance and continues to the next instance.
This implies the correctness of the algorithm. 
The running time of the algorithm is implied by Lemmas~\ref{lemma:algo-annotated-bounded-grid},~\ref{lemma:find-hor-decomp} and the fact that the algorithm presented in \cite{tree-contraction} runs in time $2^{k + o(k)}\cdot n^{\calO(1)}$. 
\end{proof}
\section{\NP-Completeness and Lower Bounds}
\label{sec:np-complete-lower-bound}

In this section, we prove that \textsc{Grid Contraction} problem is \NP-Complete.
We also argue that the dependency on the parameter in the running time of the algorithm presented in Section~\ref{sec:fpt-grid} is optimal, up to constant factors in the exponent, under a widely believed hypothesis.
We define the problems mentioned in this paragraph in the latter parts. 
Brouwer and Veldman presented a reduction from \textsc{Hypergraph
  2-Colorability} problem to $H$-\textsc{Contraction} problem \cite{brouwer1987contractibility}.
We present a reduction from \textsc{NAE-SAT} problem to
\textsc{Hypergraph 2-Colorability} problem.
We argue that the reduction used by Brouwer and Veldman can be used to
reduce the \textsc{Hypergraph 2-Colorability} problem to \textsc{Grid
  Contraction} problem.
Using these reductions and the fact there is no \emph{sub-exponential}
time algorithm for \textsc{NAE-SAT}, we obtain desired results.

We start with the definition of \textsc{Hypergraph 2-Colorability}
problem.
An hyper-edge is called monochromatic if all vertices in this edge
has the same color.
In \textsc{Hypergraph 2-Colorability} problem, an input is a hypergraph $\calH$ and the objective is
to partition $V(\calH)$ into two colors such that every edge in
$E(\calH)$ is monochromatic.
For a fixed graph $H$, the  $H$-\textsc{Contraction} problem takes a graph $G$ as an input and the objective is to determine whether $G$ can be contracted to $H$ or not.
Brouwer and Veldman proved the following result. 

\begin{proposition}[Theorem~$9$~\cite{brouwer1987contractibility}]\label{prop:C4-contr-hard} If $H$ is a connected triangle free graph other then a star then \textsc{$H$-Contraction} is \NP-Complete.
\end{proposition}

We are interested in the case when $H$ is a cycle on four vertices
(which is denoted by $C_4$).
We present the reduction that is used to prove the above proposition
in \cite{brouwer1987contractibility}.
For the sake of simplicity, we restrict the reduction to the case when
$H = C_4$.
Without loss of generality, we can assume that any instance $\calH$ of \textsc{Hypergraph 2-Colorability} contains at least two edges and has a hyper-edge which contains all vertices in $\calH$.

\vspace{0.3cm}
\noindent \textbf{Reduction-(1):} 
Given a hypergraph $\calH$ the reduction produces a graph $G$, an instance of $C_4$-\textsc{Contraction}, as follows:
\begin{itemize}
\item For every vertex $x$ in $V(\calH)$, it adds  vertex $v_x$.
Let $X$ be the set of all vertices corresponding to vertices in the hypergraph. 
\item For every hyper-edge $e$ in $E(\calH)$, it adds two vertices $e_1$ and $e_2$.
Let $E_1, E_2$ be the collection of $e_1$s and $e_2$s for all edges in the hypergraph. 
\item It adds two special vertices $v_1$ and $v_2$.
\item It adds all edges between every pair of vertices in $X$. In other words, the algorithm converts $X$ into a clique.
\item It adds all edges between every pair of vertices $(e_1, e'_2)$. In other words, the algorithm converts $G[E_1 \cup E_2]$ into a complete bipartite graph with $E_1, E_2$ as its two maximal independent sets.
\item For a vertex $v_x$ in $X$ and a vertex $e_i$ in $E_i$, if edge $e$ contains vertex $x$ then the algorithm adds edge $v_xe_i$. 
Here, $i \in \{1, 2\}$.
\item It adds edges to make $v_1$ adjacent with every vertex in $E_1$ and $v_2$ adjacent with every vertex in $E_2$.
\item It adds edge $v_1v_2$.
\end{itemize}
\vspace{0.3cm}
 
In the following claim, we argue that the above reduction can be used to prove a reduction from \textsc{Hypergraph 2-Colorability} to \textsc{Grid Contraction}. 

\begin{claim}
\label{claim:grid-cycle}
Let $G$ be the graph returned by Reduction~(1).
Graph $G$ is a \yes\ instance of $C_4$-\textsc{Contraction} if and only if $(G, |V(G)| - 4)$ is a \yes\ instance of \textsc{Grid Contraction}.
\end{claim}
\begin{proof} As $C_4$ is a $(2 \times 2)$-grid, the forward direction of the lemma is true.
Observe that the diameter of $G$ is two.
Let $k = |V(G)| - 4$.
As an edge contraction reduces the number of vertices by exactly one, if $G$ is $k$-contractible to a grid then the resulting grid has at least four vertices.
Moreover, for any graph $G'$ and an edge $e$ in it, the diameter of $G'/e$ is at most the diameter of $G$.
Hence, $G$ is $k$-contractible to a grid that has at least four vertices and has a diameter two.
Only a $(2\times 2)$-grid satisfy both of these properties.
This proves the forward direction and completes the proof of the claim.
\end{proof}

Using Proposition~\ref{prop:C4-contr-hard} and Claim~\ref{claim:grid-cycle}, we obtain the following result.

\begin{lemma}
\label{lemma:hyper-graph-grid} 
Assume that Reduction~(1) constructs graph $G$ when an input is hypergraph $\calH$. 
Then, $\calH$ is a \yes\ instance of \textsc{Hypergraph 2-Colorabiltiy} if and only if $(G, |V(G)| - 4)$ is a \yes\ instance of \textsc{Grid Contraction}. 
\end{lemma}

In the remaining section, we present a reduction from \textsc{NAE-SAT} to \textsc{Hypergraph 2-Colorability}.
In \textsc{SAT}, we are given a conjective normal formula and the
objective is to find an assignment that evaluates the formula to
\true.  
\textsc{3-SAT} is a restricted version of \textsc{SAT} in which every clause contains at most three variables.
In \textsc{NAE-SAT} variation of the problem, the objective is to find a satisfying assignment of variables such that for any clause all of its variables are not set to \true.
A simple reduction from \textsc{3-SAT} to \textsc{NAE-SAT} is as follows: Given an instance $\phi$ of \textsc{3-SAT}, add a new variable, say $x$, and replace every clause $C$ in $\phi$ by $C \land x$. 
Add a clause $(\bar{x})$ to $\phi$ to obtain an instance $\phi'$ of \textsc{NAE-SAT}.
It is easy to verify that $\phi$ is a \yes\ instance of \textsc{3-SAT} if and only if $\phi'$ is a \yes\ instance of \textsc{NAE-SAT}.
Moreover, the summation of the number of variables and the number of clauses in $\phi'$ is two more than the sum of the number of variables and the number of clauses in \textsc{NAE-SAT}.
Let $N, M$ be the number of variables and the number of a clause, respectively. 
It is know that unless Exponential Time Hypothesis (ETH) fails 
\textsc{3-SAT} problem can not be solved in time $2^{o(N + M)}$ \cite{IMPAGLIAZZO2001512}.
The above reduction implies that unless ETH fails \textsc{NAE-SAT} can not be solved in time $2^{o(N + M)}$.
We now present a reduction from \textsc{NAE-SAT} to \textsc{Hypergraph 2-Colorability}.

\vspace{0.3cm}
\noindent \textbf{Reduction-(2):} Given an instance $\phi$ of
\textsc{NAE-SAT}, the reduction algorithm constructs a hypergraph, say $\calH$, as follows:
For every variable $x$, add two vertices $x_{pos}$, $x_{neg}$. 
%Let $U_{pos}, U_{neg}$ be the collection of $x_{pos}$s and $x_{neg}$s for all variables. 
For every variable $x$, add a hyper-edge $\{x_{pos}, x_{neg}\}$.
For every clause, add a hyper-edge between the literals present in the clause.
For example, for a clause $x \land \bar{y} \land \bar{z} \land w$ add hyper-edge $\{x_{pos}, y_{neg}, z_{neg}, w_{pos}\}$.
\vspace{0.3cm}

\begin{lemma}
\label{lemma:nea-hyper-graph} 
Assume that Reduction~(2) constructs hypergraph $\calH$ when an input is formula $\phi$. 
Then, $\phi$ is a \yes\ instance of \textsc{NAE-SAT} if and only if $\calH$ is a \yes\ instance of \textsc{Hypergraph 2-Colorabiltiy}. 
\end{lemma}
\begin{proof}
In forward direction, let $\psi$ be a satisfying assignment of variables in $\phi$ such that for any clause in $\phi$, not all the literals are set to \true.
We construct a coloring function $\lambda : V(\calH) \rightarrow \{0, 1\}$ as follows:
For a variable $x$, if $\psi$ assigns $x$ to \true\ then $\lambda(x_{pos}) = 1$ and $\lambda(x_{neg}) = 0$.
If $\psi$ assigns $x$ to \false\ then $\lambda(x_{pos}) = 0$ and $\lambda(x_{neg}) = 1$.
Every edge of the type $\{x_{pos}, x_{neg}\}$,  contains a vertex which is colored $0$ and $1$.
For every edge corresponding to a clause has a vertex which is colored $1$ (as $\psi$ is a satisfying assignment) and a vertex colored $0$ (as $\psi$ does not set all literals to \true).
This implies that $\calH$ is a \yes\ instance of \textsc{Hypergraph 2-Colorability}.

In reverse direction, let $\lambda : V(\calH) \rightarrow \{0, 1\}$ be a $2$-coloring of $V(\calH)$ such that every edge contains vertices with both colors.
We construct an assignment  $\psi$ for formula $\phi$.
For a vertex $x_{pos}$, if $\lambda(x_{pos}) = 1$ then $\psi$ assigns $x$ to \true. If $\lambda(x_{pos}) = 0$ then $\psi$ assigns $x$ to \false.
We first argue that $\psi$ is a proper assignment for variables in $\phi$.
Consider a hyper-edge $\{x_{pos}, x_{neg}\}$. Since every edge has vertices with both colors, if $\lambda(x_{pos})$ then $\lambda(x_{neg}) = 0$.
This implies if $\psi$ assign $x$ to \true\ at some point, then it never assigns it to \false.
Consider a hyper-edge corresponding to a clause.
Since there is a vertex with color $1$ in this edge, $\phi$ assigned \true\ to at least one literal appearing in the clause.
Similarly, since there is a vertex with color $0$ in this edge, $\phi$ assigned \false\ to at least one literal appearing in the clause.
Hence, $\psi$ is a satisfying assignment for $\phi$ and there is no clause in $\phi$ for which $\psi$ assigns all literals to \true.
This implies that $\phi$ is a \yes\ instance of \textsc{NAE-SAT}.
\end{proof}

Observe that given an instance $\phi$ of \textsc{NAE-SAT} with $N$ variable and $M$ clauses, Reduction~(2) constructs a graph $\calH$ with $2N$ vertices and $M + N$ edges. 
This implies unless ETH fails, \textsc{Hypergraph 2-Colorability} can not be solved in time $2^{o(n' + m')}$, where $n', m'$ are the number of vertices and the number of edges in an input graph.
Note that given a hyper-graph $\calH$ on $n'$ vertices and $m'$ hyper-edges, the Reduction~(1) constructs a graph on $n' + 2m' + 2$ vertices.
This leads to the main result of this section.

\begin{theorem}\label{thm:lower-bound} \textsc{Grid Contraction} is \NP-Complete. Moreover, unless ETH fails, it can not be solved in time $2^{o(n)}$, where $n$ is the number of vertices in an input graph.
\end{theorem}

\section{Kernelization}
\label{sec:kernel}

In this section, we present a polynomial kernel for the \textsc{Grid Contraction} problem. 
In Section~\ref{sec:fpt-grid}, we reduced an instance of \textsc{Grid
  Contraction} to polynomially many instances of \textsc{Annoted
  Bounded Grid Contraction} such that the original instance is a \yes\ instance if and only one of these instances is a \yes\ instance.
One can argue that exhaustively application of Reduction Rule~\ref{rr:reduce-instance} leads to a \emph{Turing Compression}\footnote{ Please see, for example, \cite[Chapter~$22$]{fomin2019kernelization} for formal definition.} of the size $\calO(k^2)$.
We use a similar approach, but with weaker bounds, to obtain a kernel of size $\calO(k^4)$.

If the input graph is not connected then we can safely conclude that we are working with a \no\ instance. 
The following reduction rule checks two more criteria in which it is safe to return a \no\ instance.
\begin{reduction rule}\label{rr:degree-rule} For an instance $(G, k)$, if 
\begin{itemize}
\item there exists a vertex in $G$ whose degree is more than $k + 5$, or
\item there are $6k + 1$ vertices in $G$ whose degrees are more than $5$, 
\end{itemize}
then return a trivial \no\ instance. 
\end{reduction rule}

\begin{lemma}\label{lemma:rr-degree-rule-safe} Reduction Rule~\ref{rr:degree-rule} is safe.
\end{lemma}
\begin{proof} 
The maximum degree of a vertex in a grid is four. 
An edge contraction can reduce the number of vertices by one.
If a vertex has a degree more than $k + 5$ in $G$ then even after $k$ edge contractions, its degree is at least five. Hence, in this case, $G$ is not $k$-contractible to a grid.

Every vertex of degree five or more in $G$ is either in a big-witness set or it is a singleton witness set which is adjacent with a big witness set.
By Observation~\ref{obs:witness-structure-property}, there are at most $k$ big witness sets which contains at most $2k$ vertices. The big witness sets can be adjacent with at most $4k$ singleton witness set.
Hence $G$ can have at most $6k$ vertices which has a degree more than five. 
\end{proof}

We define \fbox{$k_o = (4k + 8) \cdot  (k + 1) +  1$}. 
Consider an instance $(G, k)$ on which Reduction Rule~\ref{rr:degree-rule} is not applicable.
If $G$ has at most $k_o^2 + k + 1$ vertices then we can argue that we have a kernel of the desired size.
Consider a case when $|V(G)| \ge k_o^2 + k + 1$.
We argue that in this case, if $(G, k)$ is a \yes\ instance then there exits a large grid separator in a graph $G$ (Lemma~\ref{lemma:existence-grid-sep}). 

\begin{definition}[$(p \times t)$-grid-separator] \label{def:pt-grid-sep} Consider an instance $(G, k)$ of \textsc{Grid Contraction}. 
A subset $S$ of $V(G)$ is called a $(p \times t)$-grid-separator of $G$ if it has following three properties.
\begin{itemize}
\item $G[S] = \Gamma_{p \times t}$.
\item Graph $G - S$ has exactly two connected components, say $C_1$ and $C_2$.
\item $|V(C_1)|, |V(C_2)| \ge k + 1$ and $N(C_1) = R_1, N(C_2) = R_p$, where $R_1, R_p$ are the first and last row in $G[S]$.
\end{itemize}
\end{definition}

\begin{lemma}\label{lemma:existence-grid-sep}  
Consider an instance $(G, k)$ of \textsc{Grid Contraction} such that $|V(G)| \ge k_o^2 + k + 1$. 
If $(G, k)$ is a \yes\ instance then there exists a $((4k + 6) \times t)$-grid-separator in $G$ for some integer $t$.
\end{lemma}
\begin{proof}
Assume that  $G$ is $k$-contractible to $\grid_{r \times q}$ via mapping $\psi$.
Without loss of generality, we can assume that $r \ge q$.
Since any edge contraction can reduce number of vertices in $G$ by one, the number of vertices in $\grid_{r \times q}$ is at least $k_o^2 + 1$. This implies $r \cdot q \ge k_o^2 + 1$.
Since $r \ge q$, we have $r \ge k_o$.
We show that there exists a partition of $V(G)$ into $C_1, S, C_2$ such that these sets satisfy properties mentioned in Definition~\ref{def:pt-grid-sep}. 

To satify the cardinality condition, we include all vertices in first
$(k + 1)$ many rows in $C_1$ and last $(k + 1)$ many rows in $C_2$.
Note that there are still at least $k_0 - 2(k + 1) = (4k + 6) \cdot (k + 1) + 1$ rows in the middle.
By Observation~\ref{obs:witness-structure-property}, there are at most $k$ big-witness sets.
This implies that there are at most $k$ rows in $\grid_{r \times q}$ which contain vertices corresponding to big-witness sets.
Hence, there are at least $(4k + 5) \cdot (k + 1) + 1$ rows which does not contain any big witness set.
The rows with big witness set partition the rows without any big witness set into at most $(k + 1)$ parts such that each part is connected.
At least one of these parts must have $(4k +6)$ rows. 
Hence, there exists $i_o$ in $\{k + 2, k + 3, \dots, r - (k + 2)\}$ such that no vertex in $i_o^{th}$ to $(i_o + 4k + 6)^{th}$ rows corresponds to a big witness set. 
Define $C_{1}, S, C_{2}$ as follows.
\begin{itemize}
\item $C_{1} := \{x \in V(G) |\ \psi(x) = [i, j] \text{ for some } i < i_o \text{ and } j \in [q]\}$
\item $S:= \{x \in V(G) |\ \psi(x) = [i, j] \text{ for some } i \in \{i_o, i_o + 1, \dots, i_o + 4k + 6 \} \text{ and } j \in [q]\}$.
\item $C_{2} := \{x \in V(G) |\ \psi(x) = [i, j] \text{ for some } i > i_o +4 k + 6 \text{ and } j \in [q]\}$
\end{itemize}
It is easy to verify that $C_{1}, S, C_{2}$ satisfy all the properties mentioned in Definition~\ref{def:pt-grid-sep}.
\end{proof}

In the following lemma, we argue that the existence of such a large grid
separator in a graph 
implies certain restrictions on the grid to which the graph can be contracted.

\begin{lemma}\label{lemma:large-grid-separator} 
Consider an instance $(G, k)$ of \textsc{Grid Contraction}. 
Let $S$ be a $((4k + 5) \times t)$-grid-separator of $G$.
If $G$ is $k$-contractible to a grid $\grid_{r \times q}$ then $q = t$. 
\end{lemma}
\begin{proof}
Assume that  $G$ is $k$-contractible to $\grid_{r \times q}$ via mapping $\psi$.
Note that this implies $G$ is $k$-contractible to $\grid_{q \times r}$ via mapping $\psi'$.
Rows (and corresponding witness sets) in $\grid_{r \times q}$
correspond to columns (and corresponding witness sets) in $\grid_{q \times r}$ and vice-versa.

By Observation~\ref{obs:witness-structure-property}, there are at most $k$ big-witness sets.
This implies that there are at most $2k$ rows in $G[S] \equiv \grid_{(4k + 5) \times t}$ which contain vertices which are part of big-witness sets.
Since there are $(4k + 5)$ rows, there exists $i_o$ in $\{2, 3, \dots, 4k + 5 - 2\}$ such that no vertex in $i_o^{th}$ and $(i_o + 1)^{th}$ row corresponds to a big witness set.
Let $R_{i_o} (= \{u_1, u_2, \dots, u_{t}\})$ and $R_{i_o + 1} (= \{v_1, v_2, \dots, v_{t}\})$ be the $i_o^{th}$ and $(i_o + 1)^{th}$ rows in $G[S]$.
Since no edge incident on vertices in $R_{r_o} \cup R_{r_o + 1}$ is begin contracted, we can conclude following two things:
$(a)$ For $j, j'$ in $[t]$, vertices $\psi(u_{j})$ and $\psi(u_{j'})$
(similarly, $\psi(v_{j})$ and $\psi(v_{j'})$) are adjacent with each
other if and only if $|j - j'| = 1$. 
$(b)$ For $j, j'$ in $[t]$, vertices $\psi(u_{j})$ and $\psi(v_{j'})$ are adjacent with each other if and only if $j = j'$.
This implies that $\psi(R_{i_o})$ and $\psi(R_{i_o + 1})$ are the vertices contained in two consecutive rows or columns in $\grid_{r \times q}$.
If $\psi(R_{i_o})$ and $\psi(R_{i_o + 1})$ are in two consecutive columns then we repeats the arguments with mapping $\psi'$.
Becuase of symmetry, we can assume that vertices in $\psi(R_{i_o})$ and $\psi(R_{i_o + 1})$ are in two consecutive rows in $\grid_{r \times q}$.
Let $i'$ and $i' + 1$ be the rows in $\grid_{r \times q}$ which contains vertices in $\psi(R_{i_o})$ and $\psi(R_{i_o + 1})$, respectively.
We argue that no vertices $(i')^{th}$ and $(i' + 1)^{th}$ rows  is outside $\psi(R_{i_o}) \cup \psi(R_{i_o + 1})$.
Note that $R_{i_0}$ (similarly $R_{{r_o} + 1}$) is a separators in $G$ such that there are at least two connected components of $G - R_{i_0}$ (similarly $G - R_{{r_o} + 1}$) which has at least $k + 1$ vertices.
By Observation~\ref{obs:witness-structure-property}, $\psi(R_{i_o})$ and $\psi(R_{r_o + 1})$ are two separators in $\grid_{r\times q}$.
If $\psi(R_{i_o})$ or $\psi(R_{{r_o} + 1})$ are proper subset of
vertices in $(r_o)^{th}$ or $(r_o + 1)^{th}$ row then it can not be a separator in $\grid_{r \times q}$.
This implies $\psi(R_{i_o})$ and $\psi(R_{r_o + 1})$ correspond to two rows in $\grid_{r \times q}$.
Hence the number of columns in $\grid_{r \times q}$ is equal to $|\psi(R_{i_o})| = |R_{i_o}| = q$. 
\end{proof}

We argue that if there is a large grid that is a separator in $G$ then we can safely contract two consecutive rows in this grid.

\begin{reduction rule}\label{rr:reduce-big-grid} For an instance $(G, k)$, let $S$ be a $((4k + 6)\times t)$-grid-separator of $G$ for some integer $t$.
Let $S_u (= \{u_1, u_2, \dots, u_{t}\})$ and $S_v (= \{v_1, v_2, \dots, v_{t}\})$ be two consecutive internal rows in $S$.
Let $G'$ be the graph obtained from $G$ by contracting all the edges in $\{u_jv_j |\ j \in [q]\}$.
Return instance $(G', k)$.
\end{reduction rule}

We prove that the reduction rule is safe along the same line as that of Lemma~\ref{lemma:rr-reduce-instance-safe}.

\begin{lemma}\label{lemma:rr-reduce-big-grid-safe} Reduction Rule~\ref{rr:reduce-big-grid} is safe.
\end{lemma}
\begin{proof}
Note that contracting all edges across any two consecutive rows in a grid results in another grid.

$(\Rightarrow)$
Assume $G$ is $k$-contractible to $\grid_{r \times q}$ via mapping $\psi$ for some positive integers $r, q$.
By Lemma~\ref{lemma:large-grid-separator}, $q = t$.
We argue that $G'$ is $k$-contractible to $\grid_{(r - 1) \times q}$.
As argued in the proof of Lemma~\ref{lemma:large-grid-separator}, there exists two consecutive rows $R_{i_o}$ and $R_{i_o + 1}$ in $G[S]$ such that $\psi(R_{i_o})$ and $\psi(R_{r_o + 1})$ correspond to two rows in $\grid_{r \times q}$.
Since $S_u, S_v$ are also rows in $G[S]$, we can conclude that $\psi(S_u)$ and $\psi(S_v)$ correspond to rows in $\grid_{r \times q}$.
Let $\psi(S_u)$ and $\psi(S_v)$ correspond to rows $i_1, i_2$. 
Since there are multiple edges across $S_u, S_v$, we have $|i_1 - i_2| \le 1$. 
As $G$ is $k$-contractible to $\grid_{r \times q}$,  if $|i_1 - i_2| = 1$ then $G'$ is $k$-contractible to $\grid_{(r - 1) \times q}$. 
If $|i_1 - i_2| = 0$ then as $G$ is $k$-contractible to $\grid_{r \times q}$ and $G' = G/\{v_ju_j |\ j \in [q]\}$, $G'$ is $k$-contractible to $\grid_{(r - 1) \times q}$.

$(\Leftarrow)$ Let $S_o = \{s_1, s_2, \dots, s_q\}$ be the set vertices in $G'$ which are obtained by contracting edges $u_jv_j$ in $G$. 
In other words, for $j$ in $[q]$, let $s_j$ be the new vertex added while contracting edge $u_jv_j$.
Since $S$ is a $((4k + 6) \times t)$-grid-separator in $G$, set $S' = (S \cup S_o) \setminus (S_u \cup S_v)$ is $((4k + 5) \times t)$-grid-separator in $G'$.

Assume that $G'$ is $k$-contractible to $\grid_{(r-1) \times q}$ via mapping $\phi$.
By Lemma~\ref{lemma:large-grid-separator},  $q = t$.
We argue that $G$ is $k$-contractible to $\grid_{r \times q}$.
By similar arguments as in previous part, $\phi(S_o)$ corresponding to a row, say $i_o$, in $\grid_{(r-1)\times q}$. 
As $G$ is $q$-contractible to $G'$ (which is $k$-contractible to $\grid_{(r-1)\times q}$), we know that $G$ is $(k + q)$-contractible to  $\grid_{(r-1)\times q}$. 
We define a mapping $\psi : V(G) \rightarrow \grid_{r \times q}$ corresponding to this contraction as follows: 
For every $x$ in $V(G) \setminus (S_u \cup S_v)$, define $\psi(x) = \phi(x)$ and for every $x$ in $\{ u_j, v_j\}$, define $\psi(x) = \phi(s_i)$.
%Note that $G$ is $(k + q)$-contractible to $\grid_{(r-1)\times q}$ via function $\psi$ with desirable properties.
We argue that $i_o^{th}$ row in $\grid_{(r-1)\times q}$ is partible.
For every $j$ in $[q]$, we define a partition $U_j, V_j$ of $\psi^{-1}([i_o, j])$ which satisfy all the properties mentioned in Definition~\ref{def:partible-row}. 
Since $S$ is a $(p \times t)$-grid separator of $G$, $G - S$ has two connected component  $C_1, C_2$ as specified in Definition~\ref{def:pt-grid-sep}.
Note that $G - (S_u \cup S_v)$ also has exactly two connected components, say $Y_1, Y_2$, which contain $C_1, C_2$, respectively.

For $j$ in $[q]$, let $X_j = \psi^{-1}([i_o, j])$.
As $s_j$ was present in $\phi^{-1}([i_o, j])$, vertices $u_j, v_j$ are present in $X_j$.
By the property of $\phi$, set $\phi^{-1}([i_o, j])$ is connected in $G'$. Since vertex $s_j$ is obtained from contracting edge $u_jv_j$ in $G$, graph $G[X_j]$ is connected.
Since $\phi(S_o)$ corresponds to a row of with $q$ vertices and $|S_o| = q$, vertices $u_{j'}, v_{j'}$ are present in $X_j$ if and only if $j' = j$.
In other words, $X_j \cap (S_u \cup S_v) = \{u_j, v_j\}$.
Define $U_j := (X_j \cap Y_{1}) \cup \{u_j\}$ and $V_j := (X_j \cap Y_{2}) \cup \{v_j\}$.
Since $(Y_{1}, S_u, S_v, Y_{2})$ is a partition of $V(G)$, sets $U_j, V_j$ is a non-empty partition of $X_j$.
We argue that $U_j, V_j$ satisfy all the properties in Definition~\ref{def:partible-row}.

As $N(Y_{1}) = S_u$, no vertex in $U_j \setminus \{u_j\}$ is adjacent with $S_v \cup Y_{2}$.
By similar arguments, no vertex in $V_j \setminus \{v_j\}$ is adjacent with $S_v \cup Y_{1}$. 
As $G[X_j]$ is connected and $U_j \subseteq S_u \cup Y_{1}; V_j \subseteq S_v \cup Y_{2}$, graphs $G[U_j], G[V_j]$ are connected.
Moreover, for $j' \in [q]$, sets $U_j, V_{j'}$ are adjacent if and only if $u_j, v_{j'}$ are adjacent.
Since $G[S_1\cup S_2]$ is a $(2 \times q)$-grid, $u_j$ and $v_{j'}$ are adjacent if and only if $j = j'$.
Hence $U_j$ and $V_{j'}$ are adjacent if and only if $j = j'$.
Since $U_j \subseteq X_j$ and $U_{j'} \subseteq X_{j'}$, $U_j$ and $U_{j'}$ are non adjacent if $|j - j'| > 1$.
If $|j - j'| = 1$ then $U_j, U_{j'}$ are adjacent as they contain $u_j$ and $u_{j'}$.
Hence $U_j, U_{j'}$ are adjacent if and only $|j - j'| = 1$.
By similar arguments, $V_j, V_{j'}$ are adjacent if and only if $|j - j'| = 1$.
As no vertex in $U_j$ is adjacent with $Y_{2}$ and no vertex in $V_j$ is adjacent with $Y_{1}$, we can conclude that partition $U_j, V_j$ satisfy all the properties in Definition~\ref{def:partible-row}.
Since $i_o^{th}$ row in $\grid_{(r-1)\times q}$ is partible, Lemma~\ref{lemma:partible-row} implies that $G$ is $k$-contractible to $\grid_{r \times q}$. This concludes the proof of reverse direction.
\end{proof}

The following lemma, which is analogous to
Lemma~\ref{lemma:find-hor-decomp}, is essential to argue that Reduction
Rule~\ref{rr:reduce-big-grid} can be applied in polynomial time.

\begin{lemma}\label{lemma:find-big-grid-sep} There exists an algorithm which given an instance $(G, k)$ of \textsc{Grid Contraction} and integers $p, t$ runs in polynomial time and either returns a $(p \times t)$-grid-separator of $G$ or correctly concludes that no such separator exits.
\end{lemma}
\begin{proof}
The algorithm guesses the two vertices in the first column which are
in first and last row of a potential $(p \times t)$-grid-separator of the graph.
It considers all pairs of vertices $u_1, u_p$ in $G$ which are at
distance $p$ from each other and there is a unique shortest path
between $u_1$ and $u_p$.
Let $(u_1, u_2, \dots , u_p)$ be the unique shorted path. 
For $i \in \{1, 3, \dots, p - 1\}$,  the algorithm tries to find a  subset $S_i$ of $V(G)$ which has following properties: $(a)$ $G[S_i]$ is a $(2 \times t)$ grid, $(b)$ $u_i, u_{i+1}$ are two vertices in the first column of $G[S_i]$, and $(c)$ each row in $S_i$ is a separator in $G$.
If such subset exists for every $i$ then the algorithm checks if $G - S$, where $S = \bigcup_{i}S_i$, has two connected component each with at least $k + 1$ vertices.
If it is the case then the algorithm returns $S$.
If not it moves to the next pair of vertices.
If it does not find such a set for any pair of vertices, it concludes that the graph does not contain a $(p \times t)$-grid-separator.

The algorithm returns a $(p \times t)$-grid separators only it had found one. 
Now, suppose that the graph has $(p \times t)$-grid separator $S$.
Let $G[S] = \grid_{p \times t}$, $u_1 = [1, 1]$, and $u_p = [p, 1]$.
The vertices in the first row form a unique shortest path of length $p$ between $u_1, u_p$.
Moreover, every consecutive two rows in $G[S]$ satisfies the three conditions mentioned in the above paragraph.
Hence, if the graph has a $(p \times t)$-grid separator then the algorithm returns it.
It remains to argue the running time of the algorithm.
The algorithm runs over all pairs of vertices which are at distance $p$.
It can find all such pairs exhaustively within polynomial time.
The algorithm then uses  Lemma~\ref{lemma:nr-grid-sep} to find the desired sets for $p$ pair of vertices.
The running time of the overall algorithm is implied by Lemma~\ref{lemma:nr-grid-sep} and the fact that all other steps in it can be executed in polynomial time.
\end{proof}

We are now in a position to present the main result of the section.

\begin{theorem}\label{thm:kernel-grid} \textsc{Grid Contraction} admits a kernel with $\calO(k^4)$ vertices and edges.
\end{theorem}
\begin{proof} Recall that $k_0 = (4k + 8)(k + 1) + 1$.
We assume that the input graph is connected as otherwise it is safe to conclude that we are working with a \no\ instance.
Given an instance of \textsc{Grid Contraction}, the kernelization algorithm exhaustively applies Reduction Rule~\ref{rr:degree-rule}.
Let $(G, k)$ be the resultant instance. 
If $k \le 0$ and $G$ is not a grid then the algorithm returns a \no\ instance.
If the number of vertices in the resulting instance is at most than $k_0^2 + k + 1 $ then the algorithm returns it as a kernel.
Consider a case when the number of vertices in the resulting instance is more than $k_0^2 + k + 1$.
The algorithm uses Lemma~\ref{lemma:find-big-grid-sep} to find the largest value of $t$ that  is smaller than $|V(G)|$
and there exists a $((4k + 6)\times t)$-grid-separator in $G$.
If no such $t$ exists then the algorithm returns a \no\ instance.
It then applies Reduction Rule~\ref{rr:reduce-big-grid} to obtain a smaller instance.
The algorithm repeats the process until the number of vertices in the reduces graph is at most $k_o^2 + k + 1$ or it can not find a $((4k + 6)\times t)$-grid-separator.
In the first case, it returns the reduced instance as a kernel while in another case it returns a \no\ instance.

By Lemma~\ref{lemma:existence-grid-sep}, if $(G, K)$ is a \yes\
instance there exists a $((4k + 6)\times t)$-grid-separator in $G$ for some integer $t$. This, along with Lemma~\ref{lemma:rr-degree-rule-safe} and \ref{lemma:rr-reduce-instance-safe} imply the correctness of the algorithm. 
The algorithms run in polynomial time by Lemma~\ref{lemma:find-big-grid-sep} and the fact that each application of Reduction Rule reduces the number of vertices by $t$.
As Reduction Rule~\ref{rr:reduce-big-grid} is not applicable, the reduced graph has $\calO(k^4)$ vertices.
Since Reduction Rule~\ref{rr:degree-rule} does not apply to the reduced instance, there are at most $6k$ vertices whose degree is more than $5$ and less than $k + 5$. 
The number of edges incident on these vertices is $\calO(k^2)$.
As remaining vertices have a degree at most $4$, the number of edges
in the reduced graph is $\calO(k^4)$.
This implies the reduced graph has the desired number of vertices and edges.
\end{proof}

\section{Conclusion}
\label{sec:conclusion}
In this article, we study the parameterized complexity of \textsc{Grid Contraction}.
We present an \FPT\ algorithm running that given an instance $(G, k)$ of the problem runs in time $4^{6k} \cdot n^{\mathcal{O}(1)}$ and correctly determines whether it is a \yes\ instance.
We present a notation of $r$-slab which is a generalization of a connected component of a graph.
We believe this or similar notation might be useful to get  \FPT\ or Exact Exponential Algorithms.
We prove that unless \ETH\ fails, there is no algorithm for \textsc{Grid Contraction} that runs in time $2^{o(k)} \cdot n^{\mathcal{O}(1)}$.
Finally, we prove that problem admits a kernel with $\calO(k^4)$ vertices and edges.

To the best of our knowledge, \textsc{Grid Contraction} is the only problem that admits a polynomial kernel when target graph class has unbounded path-width.
It is an interesting question to find another graph class $\calG$ such that $\calG$ has an unbounded width-parameter and $\calG$-\textsc{Contraction} admits a polynomial kernel.

\bibliographystyle{plainurl}% the recommended bibstyle

\bibliography{references.bib}

\end{document}